\documentclass{CSML}

\def\dOi{12(4:5)2016}
\lmcsheading%
{\dOi}
{1--39}
{}
{}
{Nov.~25, 2014}
{Oct.~20, 2016}
{}

\ACMCCS{[{\bf Theory of computation}]: Semantics and reasoning}
\subjclass{F.3.1, Specifying and Verifying and Reasoning about Programs}

\usepackage[english]{babel}                  
\usepackage{graphicx}                               
\usepackage{amsmath,amssymb,latexsym}               
\usepackage{epic,eepic,epsfig,graphics,psfrag}     
\usepackage{arydshln}
\usepackage{xspace}
\usepackage{mathtools}
\usepackage{stmaryrd}
\usepackage[norelsize]{algorithm2e}
\usepackage{subfigure}
\usepackage{tikz}
	\usetikzlibrary{shapes}
	\usetikzlibrary{matrix}
	\usetikzlibrary{arrows,automata}
	\usetikzlibrary{positioning}
	\usetikzlibrary{backgrounds}

\newcommand{\IGNORE}[1]{}

\newcommand{\ttt}{{\sf tt}}
\newcommand{\fff}{{\sf ff}}

\newcommand{\FV}[1]{{\sf fv}({#1})}

\newcommand{\FC}[1]{{\sf fc}({#1})}

\newcommand{\domain}[1]{{\sf dom}({#1})}

\newcommand{\NONE}{{\sf none}\xspace}
\newcommand{\SOME}{{\sf some}}

\newcommand{\NEW}[2]{(\nu {#1})\,{#2}}
\newcommand{\CASE}[4]{{\sf case}\ {#1}\ {\sf of}\ {\SOME}({#2})\!: {#3} {\ {\sf else}\ } {#4}}

\newcommand{\casE}[1] {{{\sf else}\ } {#1}}
\newcommand{\AMP}[2]{\&_{#1}(#2)}

\newcommand{\IN}[2]{{#1}?{#2}}
\newcommand{\OUT}[2]{{#1}!{#2}}
\newcommand{\NIL}{{\sf 0}}

\newcommand{\SEM}[1]{[\!\{{#1}\}\!]}
\newcommand{\DENOTATION}[1]{[\![{#1}]\!]}

\newcommand{\SYN}[1]{{\sf #1}}

\newcommand{\bisystem}{\ensuremath{P_{\Leftrightarrow}^{l}}}

\newcommand{\rightsystem}{\ensuremath{P_{\Rightarrow}^{l}}}

\newlength{\arrowlength}
\newcommand*{\strans}[1]{\xrightarrow{\mathmakebox[\arrowlength]{#1}}}

\newcommand{\dtrans}[1]{\Longrightarrow}

\newcommand{\Inference}[2]{\begin{array}{@{}c@{}}#1\\[0em]\hline\\[-0.9em]#2\\
\end{array}}

\newcommand{\Judge}[3]{{#1} \vdash {#2} \rightarrow {#3}}
\newcommand{\CONF}[3]{{#1}\!::_{#2}\!{#3}}

\newcommand{\tightoverset}[2]{%
  \mathop{#2}\limits^{\vbox to -1ex{\kern-1ex\hbox{$_#1$}\vss}}}
\definecolor{mygray}{gray}{.5}

\newcommand{\tightoversetl}[2]{%
  \mathop{#2}\limits^{\vbox to -1.3ex{\kern-1ex\hbox{$_#1$}\vss}}}
  
  \newcommand{\tightoversetv}[2]{%
  \mathop{#2}\limits^{\vbox to -.6ex{\kern-1ex\hbox{$_#1$}\vss}}}
\definecolor{mygray}{gray}{.5}

\definecolor{mygray}{gray}{.5}

\newcommand{\llt}{\leadsto}

\newcommand{\bbl}{[\hspace{-1.5pt}[}
\newcommand{\bbr}{]\hspace{-1.5pt}]}
\newcommand{\trprot}[2]{{\bbl {#1} \bbr} {#2}}
\newcommand{\hp}{\ensuremath{\varphi}}
\newcommand{\lab}[1]{\oldstylenums{#1}}
\newcommand{\trtree}[2]{{\bbl {#1} \bbr} {#2}}

\newcommand{\CASEL}[5]{{^{#1}\sf case}\ {#2}\ {\sf of}\ {\SOME}({#3})\!: {#4} {\ {\sf else}\ } {#5}}
\newcommand{\CAsel}[3]{{^{#1}\sf case}\ {#2}\ {\sf of}\ {\SOME}({#3})\!: }

\newcommand{\cost}{{\sf cost}}
\newcommand{\Names}{{\it Names}}
\newcommand{\hip}{{\sf hp}}
\newcommand{\ths}{{\sf th}}

\newcommand{\OSYN}[1] {{\sf \overline{#1}}}
\newcommand{\size}{\SYN{size}}

\newcommand{\allmodels}[1]{\ensuremath{M^{#1}}}
\newcommand{\attack}{{\sf attack}}
\newcommand{\security}{{\sf security}}
\newcommand{\minimal}{{\sf minimal}}
\newcommand{\level}{{\sf level}}

\usepackage{hyperref}

\theoremstyle{plain}\newtheorem{definition}[thm]{Definition}
\theoremstyle{plain}\newtheorem{theorem}[thm]{Theorem}
\theoremstyle{plain}\newtheorem{lemma}[thm]{Lemma}

\begin{document}

\title[Discovering, Quantifying, and Displaying Attacks]
      {Discovering, Quantifying, and Displaying Attacks\rsuper*}

\author[R.~Vigo]{Roberto Vigo}	
\address{DTU Compute, Technical University of Denmark, Richard Petersens Plads Bldg. 324, DK-2800 Kongens Lyngby, Denmark}	
\email{\{rvig, fnie, hrni\}@dtu.dk}  
\thanks{This work is supported by the IDEA4CPS project, granted by the Danish Research Foundations for Basic Research (DNRF86-10).}	

\author[F.~Nielson]{Flemming Nielson}	
\address{\vspace{-18 pt}}	

\author[H.~Riis Nielson]{Hanne Riis Nielson}	
\address{\vspace{-18 pt}}	



\keywords{Attack tree, protection analysis, Quality Calculus, Satisfiability Modulo Theories, security cost}
\titlecomment{{\lsuper*}This is a revised and extended version of~\cite{Vigo2014}}


\begin{abstract}
\noindent In the design of software and cyber-physical systems, security is often perceived as a qualitative need, but can only be attained quantitatively. Especially when distributed components are involved, it is hard to predict and confront all possible attacks. 
A main challenge in the development of complex systems is therefore to discover attacks, quantify them to comprehend their likelihood, and communicate them to non-experts for facilitating the decision process.

To address this three-sided challenge we propose a protection analysis over the Quality Calculus that (i) computes all the sets of data required by an attacker to reach a given location in a system, (ii) determines the cheapest set of such attacks for a given notion of cost, and (iii) derives an attack tree that displays the attacks graphically.

The protection analysis is first developed in a qualitative setting, and then extended to quantitative settings following an approach applicable to a great many contexts. The quantitative formulation is implemented as an optimisation problem encoded into Satisfiability Modulo Theories, allowing us to deal with complex cost structures. The usefulness of the framework is demonstrated on a national-scale authentication system, studied through a Java implementation of the framework.

\end{abstract}

\maketitle

\newpage

\tableofcontents

\newpage
\section{Introduction}\label{sec:introduction}
Our daily life is increasingly governed by software and cyber-physical systems, which are exploited in the realisation of critical infrastructure, whose security is a public concern. In the design of such systems, security is often perceived as a qualitative need, but experience tells that it can only be attained quantitatively. Especially when distributed components are involved, it is hard to predict and confront all possible attacks. Even if it were possible, it would be unrealistic to have an unlimited budget to implement all security mechanisms.

A main challenge in the development of complex systems is therefore to {\it discover} attacks, {\it quantify} them to comprehend their likelihood, and {\it communicate} them to decision-makers for facilitating deciding what countermeasures should be undertaken.

In order to address this three-sided challenge we propose a \emph{protection analysis} over the Quality Calculus~\cite{Nielson2012}. The calculus offers a succinct modelling language for specifying distributed systems with a high degree of interaction among components, which thus enjoy a highly-branching control flow. Moreover, process-algebraic languages have proven useful for describing software systems, organisations, and physical infrastructure in a uniform manner.

In the study of security, the compromise for obtaining such broad domain coverage while retaining a reasonable expressive power takes place at the level of attack definition. At a high level of abstraction, an attack can be defined as a sequence of actions undertaken by an adversary in order to make unauthorised use of a target asset. In a process-algebraic world, this necessary {\it interaction} between the adversary and the target system is construed in terms of communicating processes. Input actions on the system side can be thought of as {\it security checks} that require some information to be fulfilled. In particular, the capability of communicating over a given channel requires the knowledge of the channel itself and of the communication standard.

Hence, we shift the semantic load of secure communication onto channels, assuming that they provide given degrees of security. Whilst the correctness of such {\it secure channels}~\cite{Modersheim2009} is investigated in the realm of protocol verification, any system with a substantial need for security is likely to have standardised mechanisms to achieve various degrees of protection, and modelling them with secure channels is a coarse yet reasonable abstraction. An attack can be thus construed as a set of channels that fulfil the security checks (input actions) on a path leading to the target, i.e., a program point in the process under study.

The first task of the protection analysis, {\it discovering} all the attacks leading to a given location $l$ in a process $P$, in this setting reduces to compute all sets of channels over which communication must take place in order for $P$ to evolve to $l$. Technically, a process is translated into a set of propositional constraints that relate the knowledge of channels to communication actions and to reachability of program points. Each model of such set of constraints contains an attack. This perspective corresponds to a {\it qualitative} formulation of the analysis, in which all attacks are deemed equally interesting.

Nonetheless, secure channels allow modelling a great many domains. In IT systems, a channel can be thought of as a wired or wireless communication link over which messages are encrypted, and its knowledge mimics the knowledge of a suitable encryption key. In the physical world, a channel can represent a door, and its knowledge the ownership of the key, a pin code, or the capability of bypassing a retinal scan. This urges to acknowledge the variety of protection mechanisms and devise {\it attack metrics} by assigning costs to channels, thus quantifying the effort required to an attacker for obtaining a channel or, equivalently, the protection ensured by the security mechanism represented by a channel. As a consequence, attacks can be ordered according to their cost. Finally, a conservative approach to security demands to look for attacks that are {\it optimal} in the cost ordering (that is, minimal or maximal, according to the notion of cost we adopt -- in the following we shall focus without loss of generality on minimality).

Quantifying attacks allows to rule out those we deem unrealistic because exceeding the resources available to, or the skills of, a given attacker profile. Moreover, such a quantification allows to contrast the {\it desired security architecture} of a system, defined as a map from locations to security levels, with the {\it actual implementation of security}, defined as the protection deployed to guard a location $l$, which the analysis computes as the minimal cost required to reach $l$.

The second task of the protection analysis, {\it quantifying} the attacks leading to $l$, hence reduces to compute the costs of the attacks detected by the qualitative analysis. For the sake of performance, however, the naive approach that first computes all attacks (models) and then sorts them is not advisable. We instead implement the quest for minimal models on top of a Satisfiability Modulo Theories (SMT) solver, where minimality is sought with respect to an objective function defined on a liberal cost structure in which both symbolic and non-linearly-ordered cost sets can be represented.

The third task of the analysis, {\it displaying} attacks graphically, is achieved by backward chaining the set of constraints into which a process is translated, thus obtaining an {\it attack tree} whose root is the target location, leaves are channels, and internal nodes show how to combine sub-trees to attain the goal. Such a tree is denoted by a propositional formula whose minimal models can be computed with the SMT procedure mentioned above. Hence, our attack trees encompass both the qualitative and quantitative developments at once: a tree contains all the attacks leading to a given target, and the quest for the minimal ones can be carried out on its propositional denotation.

A Java implementation of the framework is available, which takes as input a process in the Value-Passing Quality Calculus, the location of interest, and the map from channels to costs. The tool computes the cheapest sets of atomic attacks leading to the goal, and produces an attack tree in which all attacks are summarised. The usefulness of the framework is demonstrated on the study of the NemID system, a national-scale authentication system used in Denmark to provide secure Internet communication between citizens and public institutions as well as private companies.\\

This work is a revised and extended version of~\cite{Vigo2014}. The analysis has been reformulated for highlighting its dual nature in terms of attack discovery and quantification, and has been combined with the insights of~\cite{Vigo2014a} on generating attack trees. The main novelty of the paper lies in the harmonisation of two distinct results, offering a coherent framework where the quantitative study of security and the generation of graphical models inform each other.

\subsection*{Contribution}
The novelty of the protection analysis is many-fold. First of all, the problem of inferring attacks and quantifying the protection offered by security checks is interesting {\it per se} and poorly addressed in the literature. In particular, we devise a comprehensive approach, from modelling systems and their security architecture to mechanising the verification of how this architecture has been realised in the implementation.

The extension of the analysis from qualitative to quantitative settings can be mimicked in a great many contexts, and it is seamlessly applicable to formal specifications that resort to encoding into logic formulae. In connection with this, our SMT-based solution procedure can be applied to all problems requiring to rank models of a logic formula according to given criteria.

The quest for optimal attacks can cope with arbitrary cost structures and objective functions, whose shape is only limited by the expressiveness of modern SMT solvers. Whilst our SMT approach to optimisation is not entirely novel, non-linearly-ordered and symbolic cost structures have not been addressed so far in connection with this technique. 

Finally, we overcome a number of shortcomings of existing methods for the automated generation of attack trees. These techniques typically suffer from focusing on computer networks, and thus suggest modelling languages tailored to this domain (cf. \S~\ref{sec:related-trees}). Moreover, the model-checking techniques that have been proposed recently for generating attack trees lead to an exponential explosion of the state space, limiting the applicability of automated search procedures. 

In order to overcome these drawbacks, the protection analysis derives attack trees from process-algebraic specifications in a syntax-directed fashion. As we have already put forward, process calculi have proven useful in describing a variety of domains, well-beyond computer networks. On the complexity side, being syntax-driven, static analysis often enjoys better scalability than model-checking approaches. Even though the theoretical complexity of our analysis is still exponential in the worst case, such a price depends on the shape of the process under study, and is not incurred systematically. Interestingly enough, the generation of attack trees seems to boost the performance of the approach, though it is not necessary to computing minimal attacks. In general, the interplay between static analysis and model checking had a key role in advancing the community's knowledge on the foundations of formal verification, and therefore it is promising to complement the existing studies on the attack tree generation problem with static analysis tools.

\subsection*{Organisation of the paper}
First, we present our modelling framework based on the Value-Passing Quality Calculus in \S~\ref{sec:calculus}. Besides syntax (\S~\ref{sec:calculus-syntax}) and semantics (\S~\ref{sec:calculus-semantics}) of the calculus, here we introduce our security interpretation of communication actions and the induced attacker model (\S~\ref{sec:calculus-security}). The NemID system, which will be our running example throughout the paper, is described in \S~\ref{sec:example}. 

The development of the protection analysis spans \S\S~\ref{sec:discovering}, \ref{sec:quantifying}, \ref{sec:displaying}. First, in \S~\ref{sec:discovering} we tackle the problem of discovering attacks leading to a location $l$ in a process $P$. To this end, we translate $P$ into a set of propositional constraints (\S~\ref{sec:discovering-translation}), which we then extend to encompass the attacker model (\S~\ref{sec:discovering-attacker}). The solution procedure (\S~\ref{sec:discovering-solution}) computes all the models of the constraints and from these extracts the attacks. This qualitative formulation of the analysis is demonstrated on the NemID system in \S~\ref{sec:discovering-example}, while in \S~\ref{sec:discovering-modularity} we comment upon integrating the analysis in a refinement cycle.

In \S~\ref{sec:quantifying} we extend the analysis to quantitative settings. First, we present how to formalise the intended security architecture of a system in terms of a security lattice (\S~\ref{sec:quantifying-labels}). Then, we introduce costs to channels so as to quantify attacks, and we relate the minimal costs for attacking a location to its desired security level, thus checking whether the desired security architecture is fulfilled by the implementation (\S~\ref{sec:quantifying-costs}). Our novel SMT-based approach to computing models of minimal cost is discussed in \S~\ref{sec:quantifying-solution} and demonstrated on the NemID system in \S~\ref{sec:quantifying-example}. The need for the SMT approach is motivated by the complex cost structures explained in \S~\ref{sec:quantifying-complex-costs}. We give a semantic interpretation of the attacker model in \S~\ref{sec:quantifying-restriction}.

Finally, the problem of generating attack trees is tackled in \S~\ref{sec:displaying}. First, we show how to generate trees from our sets of constraints describing how to combine channels to reach the target $l$ (\S~\ref{sec:displaying-inferring}). The procedure leads to generate a formula describing all attacks to $l$, whose optimal models can be computed thanks to the SMT procedure. We discuss the attack tree for the NemID system in \S~\ref{sec:displaying-example}.

In \S~\ref{sec:implementation} we discuss the implementation of the quantitative analysis. The solution approaches of \S\S~\ref{sec:quantifying} and \ref{sec:displaying} are compared in \S~\ref{sec:implementation-discussion}, and the proof-of-concept implementation is briefly presented in \S~\ref{sec:implementation-tool}.

Finally, \S~\ref{sec:related} surveys related work, organising them with respect to the analysis in general (\S~\ref{sec:related-analysis}) and to attack tree generation in particular (\S~\ref{sec:related-trees}). We conclude and present our outlook on future work in \S~\ref{sec:conclusion}.

\section{The Modelling Framework}\label{sec:calculus}
As specification formalism we adopt the Quality Calculus~\cite{Nielson2012}, a process calculus in the $\pi$-calculus family. In particular, the Quality Calculus introduces a new kind of input binders, termed {\it quality binders}, where a number of inputs are {\it simultaneously} active, and we can proceed as soon as some of them have been received, as dictated by a guard instrumenting the binder. Whilst these binders enhance the succinctness of highly-branching process-algebraic models, thus increasing their readability, they do not increase the expressiveness of the language with respect to $\pi$, and thus the protection analysis carries seamlessly to a variety of process calculi without such binders. Therefore, as far as this work is concerned, the reader can think of the calculus as a broadcast value-passing $\pi$-calculus enriched with quality binders.

\subsection{Syntax}\label{sec:calculus-syntax}
The syntax of the Value-Passing Quality Calculus is displayed in Table~\ref{tab:syntax}. The calculus consists of four syntactic categories: processes $P$, input binders $b$, terms $t$, and expressions $e$. A process can be prefixed by the restriction $(\nu c)$ of a name $c$, by an input binder $b$, by an output $\OUT{c}{t}$ of term $t$ on channel $c$, or be the terminated process $\NIL$, the parallel composition of two processes, a replicated process, or a $\SYN{case}$ clause.

A process $\OUT{c}{t}.P$ broadcasts term $t$ over channel $c$ and evolves to $P$: all processes ready to make an input on $c$ will synchronise and receive the message, but we proceed to $P$ even if there exists no such process (i.e., broadcast is non-blocking).

An input binder can either be a simple input $\IN{c}{x}$ binding variable $x$ to a message received over channel $c$, or a {\it quality binder} $\AMP{q}{b_1,\dots,b_n}$, where the $n$ sub-binders are {\it simultaneously} active. A quality binder is consumed as soon as enough sub-binders have been satisfied, as dictated by the quality guard $q$. A quality guard can be any Boolean predicate over the sub-binders: we use the abbreviations $\forall$ and $\exists$ for the predicates specifying that all the sub-binders or at least one sub-binder have to be satisfied before proceeding, respectively. For instance, the process $\AMP{\exists}{\IN{c_1}{x_1},\IN{c_2}{x_2}}.P$ evolves to $P$ as soon as an input is received on $c_1$, or an input is received on $c_2$.

When a quality binder is consumed, some input variables occurring in its sub-binders might have not been bound to any value, if this is allowed by the quality guard. Building on the example $\AMP{\exists}{\IN{c_1}{x_1},\IN{c_2}{x_2}}.P$, we can write:
$$
(\nu c_1)(\nu c_2)(\nu c_3)(\AMP{\exists}{\IN{c_1}{x_1},\IN{c_2}{x_2}}.P|\OUT{c_1}{c_3}.P')
$$
Such a process evolves to $P|P'$ even if no input is received on $c_2$. In order to record which inputs have been received and which have not, we always bind an input variable $x$ to an expression $e$, which is $\SOME(c)$ if $c$ is the name received by the input binding $x$, or $\NONE$ if the input is not received but we are continuing anyway. In this sense, we say that expressions represent {\it optional data}, while terms represent {\it data}, in the wake of programming languages like Standard ML~\cite{StandardML}. In order to insist on this distinction, variables $x$ are used to mark places where expressions are expected, whereas variables $y$ stand for terms, which are names (channels). Hence, in the example above variable $x_1$ is replaced by $\SOME(c_3)$ and $x_2$ is replaced by $\NONE$.

The $\SYN{case}$ clause in the syntax of processes is then used to inspect whether or not an input variable, bound to an expression, is indeed carrying data. The process $\CASE{x}{y}{P_1}{P_2}$ evolves into $P_1[c/y]$ if $x$ is bound to $\SOME(c)$; otherwise, if $x$ is bound to $\NONE$, $P_2$ is executed. Therefore, $\SYN{case}$ clauses allow detecting which condition triggered passing a binder, by detecting on which channels communication took place.
\begin{table}[t]
\caption{The syntax of the Value-Passing Quality Calculus.}
\hrule
$$
\begin{array}{lll}
P & ::= & \NIL\ \mid\ \NEW{c}{P}\ \mid\ P_1\!\mid\! P_2\ \mid\ ^{l}b.P\ \mid\ ^{l}\OUT{c}{t}.P\qquad\qquad\\
  & 	\mid & !P\ \mid \CASEL{l}{x}{y}{P_1}{P_2}\qquad\qquad \mbox{\small  (Process)}\\[1ex]
b & ::= & \IN{c}{x}\ \mid\ \AMP{q}{b_1, \cdots , b_n}\hfill \mbox{\small  (Binder)}\\[1ex]
t & ::= & c\ \mid\ y \hfill \mbox{\small  (Term)}\\[1ex]
e & ::= & x\ \mid\ \SOME(t)\ \mid\ \NONE \hfill \mbox{\small  (Expression)}
\end{array}
$$
\hrule
\label{tab:syntax}
\end{table}

In order to identify locations unambiguously, the syntax instruments each program point with a {\it unique} label $l \in \mathcal{L}$. This is the case of input binders, outputs, and case clauses in Table~\ref{tab:syntax}. In the following, we will denote labels with the numerals $\lab{1}, \lab{2}, \dots$. We say that a label $l$ is reached in an actual execution when the sub-process following $l$ is ready to execute. 

As usual, in the following we consider closed processes (no free variable), and we make the assumption that variables and names are bound exactly once, so as to simplify the technicalities without impairing the expressiveness of the framework.\\

The syntax presented above defines a fragment of the full calculus developed in~\cite{Nielson2012}. First, we have described a value-passing calculus, as channels are syntactically restricted to names as opposed to terms; in particular, no variable $y$ is allowed in a channel position. Second, terms are pure, that is, only names and variables are allowed (no function application). Last, the calculus is limited to test whether or not a given input variable carries data, but cannot test which data it is possibly carrying. We deem that these simplifications match with the abstraction level of the analysis, where secure channels are assumed that have to be established by executing given security protocols, and thus enjoy known properties. For the very same reason we opted for pure terms, since equational reasoning is usually exploited in process calculi to model cryptographic primitives (a Quality Calculus with such a feature is presented in~\cite{Vigo2013}). 

Ultimately, resorting to a value-passing calculus is to be understood as a choice that allows to produce attack trees of reasonable size, and not as a technical shortcoming. The comparison here is to be made between security protocols, often studied with expressive calculi (cf. \S~\ref{sec:related-analysis}), and distributed systems. Security protocols can be understood as basic building blocks that we would like to prove flawless; to this end, modelling the structure of messages and the functions operating over them is necessary (and often not enough). On the other hand, such a low-level representation seems instead not suitable to modelling complex systems and to support analyses whose results can be computable in practice and presented in a human-readable format. In this sense, the coarse abstraction we adopt with flat channels is the price to pay for scaling up from analysing protocols to analysing complex systems, where a flat name can represent an entire sub-system (software, physical, cyber-physical) whose analysis is condensed in the price/security guarantee attached to it.

In this light, a reinforced gate can be represented with a channel $g$, and its cost be the minimal grams of dynamite required to blow it up, or the average price of the tools needed to force it, or a synthesis of both, depending on the attacker we want to consider. Similarly, a secure communication link providing secrecy can be represented with a channel $s$, and its cost be the length of the cryptographic key used for encrypting information sent over $s$, or the trust we put in the encryption protocol used to implement $s$. More complex cost notions can be used to describe security properties of cyber-physical devices, where for example the battery level can dictate the security degree a sensor can afford. While it is true that part of the burden is moved onto the definition of suitable cost structures (some solutions particularly useful to facilitate the task in case of cyber-physical systems are advanced in \S~\ref{sec:quantifying-complex-costs}), the framework offers a clear separation of concerns between modelling a complex system at an appropriate detail level, specifying the security contracts of its component, and computing a weak-path analysis.

An extension of the framework that encompasses the key features of the full calculus of~\cite{Nielson2012} is briefly commented upon in Appendix~\ref{sec:fo-trees}, and is a starting point to delve into more specialised notions of costs, where the analysis on channels is combined with an analysis on the structure of messages.

\subsection{Semantics}\label{sec:calculus-semantics}
The semantics of the Value-Passing Quality Calculus consists of a set of transition rules parametrised on the structural congruence relation of Table~\ref{tab:congruence}. In the congruence, we denote the free names in a process $P$ as $\FC{P}$, where $\NEW{c}{P}$ binds the name $c$ in process $P$. Moreover, the congruence holds for contexts $C$ defined as
$$
C ::= [\ ]\ \mid\ \NEW{c}{C}\ \mid\ C\!\mid\!P\ \mid\ P\!\mid\!C
$$
As usual in $\pi$-like calculi, processes are congruent under $\alpha$-conversion, and renaming is enforced whenever needed in order to avoid accidental capture of names during substitution.

The semantics of Table~\ref{tab:semantics} is based on the transition relation $P \Longrightarrow P'$, which is enabled by combining a local transition $\strans{\lambda}$ with the structural congruence.
\begin{table}[t]
\caption{The structural congruence of the Value-Passing Quality-Calculus.}
\hrule
$$
\begin{array}{c}
\begin{array}{ccc}
P \equiv P
&
P\!\mid\! \NIL \equiv P
&
P_1\!\mid\! P_2 \equiv P_2\!\mid\! P_1\\[2ex]
P_1\!\mid\! (P_2\!\mid\! P_3) \equiv (P_1\!\mid\! P_2)\!\mid\! P_3
\qquad&\qquad
!P \equiv P | !P
\qquad&\qquad
\NEW{c_1}{\NEW{c_2}{P}} \equiv \NEW{c_2}{\NEW{c_1}{P}}
\end{array}\\[4ex]
\begin{array}{lr}
\NEW{c}{P} \equiv P \quad\hbox{if } c\notin\FC{P} 
\quad\qquad&\quad\quad
\NEW{c}{(P_1\mid P_2)} \equiv (\NEW{c}{P_1})\mid P_2 \quad\hbox{if } c\notin \FC{P_2}
\end{array}\\[2ex]
\begin{array}{ccc}
\Inference{P_1 \equiv P_2}{P_2 \equiv P_1}
\qquad&\qquad
\Inference{P_1 \equiv P_2 \qquad P_2 \equiv P_3}{P_1 \equiv P_3}
\qquad&\qquad
\Inference{P_1 \equiv P_2}{C[P_1] \equiv C[P_2]}
\end{array}
\end{array}
$$
\hrule
\label{tab:congruence}
\end{table}

The semantics models asynchronous broadcast communication, and makes use of a label $\lambda ::= \tau\ |\ \OUT{c_1}{c_2}$ to record whether or not a broadcast $\OUT{c_1}{c_2}$ is available in the system. If not, label $\tau$ is used to denote a silent action. Rule (Brd) states that an output is a non-blocking action, and when it is performed the broadcast is recorded on the arrow. When a process guarded by a binder receives an output, the broadcast remains available to other processes able to input. This behaviour is encoded in rules (In-ff) and (In-tt), where we distinguish the case in which a binder has not received enough inputs, and thus keeps waiting, from the case in which a binder is satisfied and thus the computation may proceed, applying the substitution induced by the received communication to the continuation process. These rules rely on two auxiliary relations, one defining how an output affects a binder, and one describing when a binder is satisfied (enough inputs have been received), displayed in the third and fourth section of Table~\ref{tab:semantics}, respectively.

We write $\Judge{\OUT{c_1}{c_2}}{b}{b'}$ to denote that the availability of a broadcast makes binder $b$ evolve into binder $b'$. When a broadcast name $c_2$ is available on the channel over which an input is listening, the input variable $x$ is bound to the expression $\SOME(c_2)$, marking that something has been received. Otherwise, the input is left syntactically unchanged, that is, it keeps waiting. Technically, the syntax of binders is extended to include substitutions:
$$
b ::= \cdots\, |\, [\SOME(c)/x]
$$
This behaviour is seamlessly embedded into quality binders, where a single output can affect a number of sub-binders, due to the broadcast paradigm and to the intended semantics of the quality binder, according to which the sub-binders are {\it simultaneously} active. 

As for evaluating whether or not enough inputs have been received, the relation \mbox{$\CONF{b}{r}{\theta}$} defines a Boolean interpretation of binders. An input evaluates to $\theta = [\NONE/x]$ with $r = \fff$, for it must be performed before continuing with the computation; a substitution evaluates to itself with $r = \ttt$, since it stands for a received input. The substitution related to a quality binder is obtained by composing the substitutions given by its sub-binders, while its Boolean interpretation is obtained applying the quality guard, which is a Boolean predicate, to the Boolean values representing the status of the sub-binders. This is denoted $\SEM{q}(r_1,\dots,r_n)$, and corresponds to the Boolean $\lor$ for the existential guard and to the Boolean $\land$ for the universal guard.
\begin{table*}[t!]
\caption{A broadcast value-passing semantics with replication.}
\hrule
$$
\Inference
{P_1 \equiv \NEW{\overrightarrow{c}}{P_2}\quad P_2 \strans{\lambda} P_3}
{P_1 \Longrightarrow P_3}\quad \mbox{\small (Sys)}
$$
\textcolor{gray}{\hrule}
$$
\begin{array}{l}
%
%
^{l}\OUT{c_1}{c_2}.P \strans{\OUT{c_1}{c_2}}P\hfill \mbox{\small (Brd)}
\\[3ex]
%
%
\Inference{
P_1 \strans{\OUT{c_1}{c_2}}P'_1 \quad
\OUT{c_1}{c_2} \vdash b \rightarrow b'\quad
\CONF{b'}{\fff}{\theta}
}{
P_1 \mid {^l}b.P_2 \strans{\OUT{c_1}{c_2}} P'_1 \mid {^l}b'.P_2
}\hfill \mbox{\small (In-ff)}\\[4ex]
\Inference{
P_1 \strans{\OUT{c_1}{c_2}}P'_1 \quad
\OUT{c_1}{c_2} \vdash b \rightarrow b'\quad
\CONF{b'}{\ttt}{\theta}
}{
P_1 \mid {^l}b.P_2 \strans{\OUT{c_1}{c_2}} P'_1 \mid P_2\theta
}\hfill \mbox{\small (In-tt)}\\[4ex]
\CASEL{l}{\SOME(c)}{y}{P_1}{P_2} \strans{\tau} P_1[c/y]\qquad\qquad\hfill \mbox{\small (Then)}\\[3ex]
\CASEL{l}{\NONE}{y}{P_1}{P_2} \strans{\tau} P_2\hfill \mbox{\small (Else)}\\[2ex]
%
\Inference{P_1\strans{\tau}P'_1}{P_1|P_2\strans{\tau}P'_1|P_2}\hfill \mbox{\small (Par-tau)}\\[4ex]
%
\Inference{P_1\strans{\OUT{c_1}{c_2}}P'_1}{P_1|P_2\strans{\OUT{c_1}{c_2}}P'_1|P_2}\ \mbox{ if } P_2 =\, !P'_2 \lor P_2 = {^l}\OUT{c'_1}{c'_2}P'_2\hfill \mbox{\small (Par-brd)}
\end{array}
$$
\textcolor{gray}{\hrule}
$$
\begin{array}{c}
\\[-.5ex]
\Judge{\OUT{c_1}{c_2}}{\IN{c_1}{x}}{[\SOME(c_2)/x]}
\qquad\qquad
\Judge{\OUT{c_1}{c_2}}{\IN{c_3}{x}}{\IN{c_3}{x}}\ \mbox{ if } c_1 \neq c_3\\[2ex]
%
%
\Inference{
\Judge{\OUT{c_1}{c_2}}{b_1}{b'_1}\quad \cdots\quad \Judge{\OUT{c_1}{c_2}}{b_n}{b'_n}
}{
\Judge{\OUT{c_1}{c_2}}{\AMP{q}{b_1,\dots,b_n}}{\AMP{q}{b'_1,\dots,b'_n}}
}
\end{array}
$$
\textcolor{mygray}{\hrule}
$$
\begin{array}{c}
\\[-2ex]
\CONF{\IN{t}{x}}{\fff}{[\NONE/x]}
\qquad\qquad
\CONF{[{\SOME}(c) / x]}{\ttt}{[{\SOME}(c) / x]}
\\[2ex]
\Inference{
\CONF{b_1}{r_1}{\theta_1}\quad
\cdots \quad\CONF{b_n}{r_n}{\theta_n}}{
\CONF{\AMP{q}{b_1,\cdots,b_n}}{r}{\theta_n \circ\cdots \circ \theta_1}
}
 \hbox{ where } r= \SEM{q}(r_1,\cdots,r_n)
\end{array}
$$
\hrule
\label{tab:semantics}
\end{table*}

It is worthwhile observing that substitutions are applied directly by the semantics, and since we consider closed processes, whenever an output $\OUT{c_1}{t}$ is ready to execute, $t$ must have been replaced by a name $c_2$ (cf. rule (Brd)). For this very same reason the evaluation of $\CASE{x}{y}{P_1}{P_2}$ is straightforward, as when the clause is reached in an actual execution the variable $x$ must have been replaced with a constant expression.

It is worthwhile observing that no rule for transition under restriction is provided, therefore rule (Sys) must be applied to pull restrictions to the outer-most level. Here $(\nu\overrightarrow{c})$ denotes the restriction of a list of names. 

Rules (Par-) take care of interleaving. In particular, rule (Par-brd) takes care of interleaving broadcast, and applies in two situations.

First, if the parallel component $P_2$ is a replicated process, the broadcast has precedence over any action of instances of $P_2$. This means that when $P_2$ has the form $!\IN{c_1}{x}.P$, the number of input performed by its replicated instances will depend on the unfolding performed previously using the structural congruence as per rule (Sys). Note that this ``freedom'' issue would not arise if the semantics made use of guarded recursion. Mind to observe, however, that the encoding of replication into recursion would lead to the same problem, since replication is essentially an unguarded recursion. 
%

Second, if the parallel component $P_2$ is broadcast-prefixed, its broadcast is delayed until the current one has exhausted its synchronisation opportunities. If none of the above applies, then either a synchronisation rule (In-) or an interleaving with a silent action (Par-tau) must take place.

\subsection{Security model}\label{sec:calculus-security}
On top of the standard operational semantics of the calculus, our process-algebraic specifications rely on a security interpretation of communication actions. As binders are blocking actions, they can be thought of as security checks that require the knowledge of some information to be fulfilled. This knowledge is abstracted here by resorting to the notion of {\it secure channel} and disregarding the messages actually communicated, in line with the value-passing nature of the calculus. This idea seamlessly applies to simple inputs and is refined by quality binders: the existential quality guard $\exists$ describes scenarios in which different ways of fulfilling a check are available, e.g., different ways of proving one's identity. In contrast to this, the universal quality guard $\forall$ describes checks that require a number of sub-conditions to be met at the same time, and can be used to refine a security mechanism in terms of sub-checks.

On the other hand, an output represents the satisfaction of the security check specified by the corresponding channel. As the semantics is broadcast, all the input waiting on the given channel will be satisfied, that is, all the pending instances of the security check will be fulfilled and the system will proceed until another blocking check is met. As a consequence, if the adversary can trigger a system component to make an output on a given channel $c$, it is as if $c$ were under the control of the attacker, for all inputs on $c$ in other components would be satisfied.

Finally, $\SYN{case}$ clauses allow to inspect how a given security check, that is, a preceding binder, has been satisfied, by inspecting which input variables are bound to some values, that is, on which channels the communication took place.

As a result, an attack can be construed as a set of channels, namely, those channels needed to fulfil the security checks on a path to the target. We assume that a system $P$ is deployed in a hostile environment, simulated by an adversary process $Q$ running in parallel with $P$. The ultimate aim of the protection analysis is to compute what channels $Q$ has to communicate over in order to drive $P$ to a given location $l$, i.e., $P|Q \Longrightarrow^* C[^{l}P']$, where $C[^{l}P']$ denotes a sub-process of $P|Q$ that has reached label $l$ and $\Longrightarrow^*$ denotes the reflexive and transitive closure of $\Longrightarrow$.

In order to compute the channels that $Q$ needs to reach a program point $l$, this security model is translated into an attacker model where $Q$ can {\it guess} every required channel, that is, whenever a security check on the way to $l$ cannot be avoided, $Q$ can fulfil it.

\section{The Danish NemID System}\label{sec:example}
Let us introduce now an example that will be developed throughout the paper. NemID (literally: EasyID)\footnote{\url{https://www.nemid.nu/dk-en/}} is an asymmetric cryptography-based log-in solution used by public institutions as well as private companies in Denmark to provide on-line services, used by virtually every person who resides in the country. Most service providers rely on a browser-based log-in application through which their customers can be authenticated (a Java applet or a JavaScript program -- in the following simply ``the applet''). For technological and historical reasons, the applet allows proving one's identity with various sets of credentials. In particular, private citizens can log-in with their social security number, password, and a one-time password, or by exhibiting an X.509-based certificate. Moreover, on mobile platforms that do not support Java, a user is authenticated through a classic id-password scheme.

We focus on a formalisation of the system such that authentication does not take place unless some interactions with the environment take place. In other words, we specify the security checks but not their fulfilment by the user who is supposed to be logging in, so as to rule out the legal way to authenticating and focus on malicious behaviours.

The system is modelled in the Value-Passing Quality Calculus as follows:
$$
{\it NemID} \triangleq (\nu\SYN{cert})\dots(\nu\SYN{access})(!{\it Login}\ |\ !{\it Applet}\ |\ !{\it Mobile})
$$
$$
\begin{array}{l}
{\it Applet} \triangleq \\
 \quad^{\lab{1}}\AMP{\exists}{\IN{\SYN{cert}}{x_{\SYN{cert}}},\AMP{\forall}{\IN{\SYN{id}}{x_{\SYN{id}}}, \IN{\SYN{pwd}}{x_{\SYN{pwd}}}, \IN{\SYN{otp}}{x_{\SYN{otp}}}}}.\\
 \quad\CAsel{\lab{2}}{x_{\SYN{cert}}}{y_{\SYN{cert}}}{^{\lab{3}}\OUT{\SYN{login}}{\SYN{ok}}}\ \SYN{else}\\
\quad \qquad \CAsel{\lab{4}}{x_{\SYN{id}}}{y_{\SYN{id}}}\\
\quad \qquad\qquad \CAsel{\lab{5}}{x_{\SYN{pwd}}}{y_{\SYN{pwd}}}\\
\quad \qquad\qquad\qquad \CASEL{\lab{6}}{x_{\SYN{otp}}}{y_{\SYN{otp}}}{^{\lab{7}}\OUT{\SYN{login}}{\SYN{ok}}}{\NIL}\\
\quad \qquad\qquad \casE{\NIL}\\
\quad \qquad \casE{\NIL}
\end{array}
$$
$$
\begin{array}{l}
{\it Mobile} \triangleq\ ^{\lab{8}}\AMP{\forall}{\IN{\SYN{id}}{x'_{\SYN{id}}}, \IN{\SYN{pin}}{x_{\SYN{pin}}}}.\\
\qquad\qquad\quad\CAsel{\lab{9}}{x'_{\SYN{id}}}{y'_{\SYN{id}}}\\
\qquad\qquad\qquad \CASEL{\lab{10}}{x_{\SYN{pin}}}{y_{\SYN{pin}}}{^{\lab{11}}\OUT{\SYN{login}}{\SYN{ok}}}{\NIL}\\ 
\qquad\qquad\quad \casE{\NIL}\\[2ex]
{\it Login} \triangleq\ ^{\lab{12}}\IN{\SYN{login}}{x}.^{\lab{13}}\OUT{\SYN{access}}{\SYN{ok}}
\end{array}
$$
The system consists of three processes running in parallel an unbounded number of times. For the sake of brevity, we have omitted to list all the restrictions in front of the parallel components, that involve all the names occurring in the three processes.

Process {\it Login} is in charge of granting access to the system: whenever a user is authenticated via the applet or a mobile app, an output on channel $\SYN{login}$ is triggered that can be received at label $\lab{12}$ leading to the output at label $\lab{13}$, which simulates a successful authentication.

Process {\it Applet} models the applet-based login solution, where login is granted (simulated by the outputs at labels $\lab{3}$ and $\lab{7}$) whenever the user exhibits a valid certificate or the required triple of credentials. The quality binder at label $\lab{1}$ implements such a security check: in order to pass the binder, either ($\exists$) a certificate has to be provided, simulated by the first sub-binder, or three inputs have to be received ($\forall$), mimicking an id ($\SYN{id}$), a password ($\SYN{pwd}$), and a one-time password ($\SYN{otp}$). Observe how $\SYN{case}$ constructs are used to determine what combination allowed passing the binder: at label $\lab{2}$ we check whether the certificate is received, and if this is not the case then we check that the other condition is fulfilled. The main abstraction of our approach takes place at this level, as we can only test whether something is received on a given channel, but we cannot inspect the content of what is received. In other words, the knowledge of channel $\SYN{cert}$ mimics the capability of producing a valid certificate, and thus we say that the semantic load of the communication protocol is shifted onto the notion of secure channel. Observe that this perspective seamlessly allows reasoning about the cost of attacking the system: to communicate over $\SYN{cert}$, an adversary has to get hold of a valid certificate, e.g., bribing someone or brute-forcing a cryptographic scheme, and this might prove more expensive than guessing the triple of credentials necessary to achieve authentication along the alternative path, as we shall conclude in \S~\ref{sec:quantifying-example}.

Finally, process {\it Mobile} describes the intended behaviour of the mobile login solution developed by some authorities (e.g., banks, public electronic mail system), where an id and a password or pin have to be provided upon login.

\section{Discovering Attacks}\label{sec:discovering}
The first challenge we tackle is discovering all the attacks leading to a given target. Given a process $P$ modelling the system of interest and a label $l$ in $P$ identifying the target of the attack, the task of the analysis is to find all the sets of channels fulfilling the inputs occurring in $P$ on the paths to $l$. In other words, we look for the knowledge that allows an adversary $Q$ to drive $P$ to reach $l$.

To this end, $P$ is translated into a set of propositional formulae (\S~\ref{sec:discovering-translation}), termed {\it flow constraints}, describing how the knowledge of channels relates to reachability of locations, and how given a set of channels some other channels can be derived by the attacker. Then, the constraints are extended to consider the capabilities of the attacker (\S~\ref{sec:discovering-attacker}). Each model of the final set of constraints $\bisystem$ contains a set of channels that under-approximates the knowledge required to reach $l$ (\S~\ref{sec:discovering-solution}).

\subsection{From processes to flow constraints}\label{sec:discovering-translation}
The call $\trprot{P}{\ttt}$ of the recursive function $\trprot{\cdot}{\cdot}$, defined in Table~\ref{tab:translation-clauses}, translates a process $P$ into a set of constraints of the form  $\varphi \llt \overline{p}$, where $\varphi$ is a propositional formula and $\overline{p}$ a positive literal. The intended semantics of a constraint states that if $Q$ knows (enough information to satisfy) $\varphi$, then $Q$ knows $p$, i.e., $\overline{p} = \ttt$. As we shall see below, the antecedent $\varphi$ accounts for the checks made on the path leading to disclosing $p$, namely input binders and $\SYN{case}$ clauses. The consequent $\overline{p}$ can either stand for a channel literal $\overline{c}$, meaning that $Q$ controls $c$, an input-variable literal $\overline{x}$, meaning that the attacker can satisfy the related input (i.e., $x = \SOME(c)$), or a label literal $\overline{l}$, meaning that $l$ is reached.

At each step of the evaluation, the first parameter of $\trprot{\cdot}{\cdot}$ corresponds to the sub-process of $P$ that has still to be translated, while the second parameter is a logic formula, intuitively carrying the hypothesis on the knowledge $Q$ needs to reach the current point in $P$. The translation function is structurally defined over processes as explained below.

If $P$ is $\NIL$, then there is no location to be attained and thus no constraint is produced. If $P =\ !P'$ or $P = \NEW{c}{P'}$, then it spontaneously evolves to $P'$, hence $Q$ does not need any knowledge to reach $P'$ and gains no knowledge since no communication is performed. A parallel composition is translated taking the union of the sets into which the components are translated.

Communication-related actions have instead an impact on the knowledge of $Q$: inputs represent checks that require knowledge, outputs fulfil those checks, and $\SYN{case}$ clauses determine the control flow. Assume that an action $\pi$ is reached in the translation under hypothesis $\varphi$, that is, $\trprot{^{l}\pi.P'}{\varphi}$. Then, a constraint $\varphi \llt \overline{l}$ is generated: if the attacker fulfils the security checks on a path to $l$, then $l$ is reached. In the logic interpretation of the constraints, this happens when $\varphi$ evaluates to $\ttt$ under a model given by the knowledge of the attacker, forcing $\overline{l}$ to be $\ttt$ (as standard implication would do). Moreover, the nature of action $\pi$ determines whether or not other constraints are produced and how the translation proceeds.

Consider a simple input $^l\IN{c}{x}.P'$: whenever the action is consumed, it must be that the attacker controls the communication channel $c$, hence we translate $P'$ under the hypothesis $\varphi \land \overline{c}$. Moreover, if the input is consumed, then $x$ must be bound to $\SOME(c')$, hence we produce a constraint $(\varphi \land \overline{c}) \llt \overline{x}$. These two steps respectively accommodate the hypothesis we need for passing a binder (success condition), and the conclusion we can establish whenever a binder is passed (strongest post-condition). In Table~\ref{tab:translation-clauses} functions $\hip$ and $\ths$ take care of formalising this intuition, that seamlessly applies to quality binders, where the hypothesis is augmented accounting for the combinations of inputs that satisfy the binder, as dictated by the quality guard $q$. The last section of Table~\ref{tab:translation-clauses} shows two cases for $q$, but any Boolean predicate can be used.

The execution of an output $^l\OUT{c}{t}$ satisfies all the security checks represented by inputs waiting on $c$. Therefore, if $Q$ can trigger such output, it obtains the knowledge related to $c$ without having to know the channel directly, and thus a constraint $\hp \llt \overline{c}$ is generated. It is worthwhile observing that this behaviour is justified by the broadcast semantics, and by the fact that the calculus is limited to testing  whether or not something has been received over a given channel, shifting the semantic load on the notion of secure channel. Moreover, note that the asymmetry between input and output is due to the fact that outputs are non-blocking.

A $\SYN{case}$ construct is translated by taking the union of the constraints into which the two branches are translated: as the check is governed by the content of the $\SYN{case}$ variable $x$, we record that the $\SYN{then}$ branch is followed only when $x$ is bound to $\SOME(c)$ by adding a literal $\overline{x}$ to the hypothesis, as we do for inputs, and we add $\neg \overline{x}$ if the $\SYN{else}$ branch is followed.\\
\begin{table*}[t]
\caption{The translation $\trprot{P}{\ttt}$ from processes to flow constraints.}
\hrule
$$
\begin{array}{l}
\begin{array}{lll}
\trprot{\NIL}{\hp} 			& = & \emptyset\\[.5ex]
\trprot{!P}{\hp}				& = & \trprot{P}{\hp}\\[.5ex]
\trprot{P_1|P_2}{\hp} 		& = & \trprot{P_1}{\hp} \cup \trprot{P_2}{\hp}\\[.5ex]
\trprot{\NEW{c}{P}}{\hp}		& = & \trprot{P}{\hp}\\[.5ex]
\trprot{^lb.P}{\hp}			& = & \trprot{P}{(\hp \land \hip(b))} \cup \ths(\hp,b) \cup \{\hp \llt \overline{l}\}\\[.5ex]
\trprot{^l\OUT{c}{t}.P}{\hp}	& = & \trprot{P}{\hp} \cup \{\hp \llt \overline{c}\} \cup \{\hp \llt \overline{l}\}\\[1ex]
\end{array}\\
\trprot{\CASEL{l}{x}{y}{P_1}{P_2}}{\hp}	= \trprot{P_1}{(\hp \land \overline{x})} \cup \trprot{P_2}{(\hp \land \neg \overline{x})} \cup \{\hp \llt \overline{l}\}
\end{array}
$$
\textcolor{mygray}{\hrule}
$$
\begin{array}{ll}
\begin{array}{l}
\hip(\IN{c}{x}) = \overline{c}\\[1ex]
\ths(\hp,\IN{c}{x}) = \{(\hp \land \overline{c}) \llt \overline{x}\}
\end{array}
&\quad
\begin{array}{l}
\hip(\AMP{q}{b_1,\dots, b_n}) = \SEM{q}(\hip(b_1),\dots,\hip(b_n))\\[1ex]
\ths(\hp,\AMP{q}{b_1,\dots, b_n}) = \bigcup_{i=1}^{n}\ths(\hp,b_i)
\end{array}
\end{array}
$$
\textcolor{mygray}{\hrule}
$$
\begin{array}{c}
\SEM{\forall}(\overline{c_1},\dots,\overline{c_n}) = \bigwedge_{i=1}^n \overline{c_i}
\qquad\qquad\qquad\qquad
\SEM{\exists}(\overline{c_1},\dots,\overline{c_n}) = \bigvee_{i=1}^n \overline{c_i}
\end{array}
$$
\hrule
\label{tab:translation-clauses}
\end{table*}

The set of constraints $\trprot{P}{\ttt}$ computed according to Table~\ref{tab:translation-clauses} can be normalised so as to produce a compact representation of $P$. Whenever two rules $\varphi \llt \overline{p}$ and $\varphi' \llt \overline{p}$ are in $\trprot{P}{\ttt}$, they are replaced with a single rule $(\varphi \lor \varphi') \llt \overline{p}$. This simplification is intuitively sound for if $\varphi$ leads to obtain $p$ and $\varphi'$ leads to obtain $p$, then $p$ is available to the attacker under the condition that $\varphi \lor \varphi'$ is known. In the following, we assume to deal with sets of constraints in such format.
	
\subsection{Modelling the attacker}\label{sec:discovering-attacker}
A rule $\varphi \llt \overline{p}$ in $\trprot{P}{\ttt}$ describes how $Q$ can attain $p$ playing according to the rules of the system, namely fulfilling the checks described by $\varphi$. Nonetheless, when $p$ is a channel, an attacker can always try to obtain it directly, for instance guessing some cryptographic keys or bursting a gate. In order to account for this possibility, we enrich each rule $\varphi \llt \overline{c}$ by replacing the antecedent with the disjunction $g_c \lor \varphi$, where literal $g_c$ (for ``guess $c$'') represents the possibility of learning $c$ directly. For each channel $c$ such that no rule $\varphi \llt \overline{c}$ is in $\trprot{P}{\ttt}$, we add to $\trprot{P}{\ttt}$ a constraint $g_c \llt \overline{c}$, expressing that $Q$ has no option but guessing the channel.

Finally, observe that having added the literals $g_c$, which tell how the attacker can get hold of a channel in any other way than those legal in $P$, interpreting the relation $\llt$ as the propositional bi-implication (equivalence) $\Leftrightarrow$ preserves the minimal models of the system of constraints. A constraint $\varphi \lor g_c \Leftrightarrow \overline{c}$ states that $c$ is only obtained by guessing or by making $P$ disclose it and, on the other hand, that if $c$ is known to the attacker it must be because they have guessed it or because they made $P$ disclose it.

In the following, given a label $l$ of interest, we shall write $\bisystem$ to denote the conjunction of constraints $\trprot{P}{\ttt}$ which have undergone the transformations mentioned above. In particular, in $\bisystem$
\begin{itemize}
\item $\overline{l}$ is a fact, expressing the query we want to study;
\item there is exactly one conjunct $\varphi \Leftrightarrow \overline{l}$;
\item for each channel $c$ occurring in $P$, there is exactly one conjunct $\varphi \Leftrightarrow \overline{c}$, and $\varphi$ has the form $g_c \lor \varphi'$ where $g_c$ does not occur elsewhere in $\bisystem$.
\end{itemize}
Intuitively, we assume that $l$ is reached and we look for the consequences in terms of truth values of channel literals (i.e., we look for \emph{implicants} of $\overline{l}$~\cite{DeMicheli}). In the following, we present a SAT-based solution to this problem.

\subsection{A SAT-based solution technique}\label{sec:discovering-solution}
A model $\mu$ of $\bisystem$ is a propositional assignment to the literals occurring in $\bisystem$ such that the formula evaluates to $\ttt$ (true). Such an assignment can be represented as a function mapping the literals in $\bisystem$, denoted $\domain{\mu}$, to truth values $\{\fff,\ttt\}$. Now, denoted by $\Names$ the set of names (channels) occurring in the process $P$ under study, then the set
$$
\attack(\mu) = \{\overline{c} \in \domain{\mu}\ |\ c \in \Names\, \land\, \mu(g_c) = \ttt \}
$$
identifies a set of channels that, if guessed, satisfies the constraints in $\bisystem$. In the semantics, $\attack(\mu)$ under-approximates a set of channels that fulfil the security checks on a path to $l$ in $P$, i.e., a way for $Q$ to drive $P$ to $l$. Denoted by $\allmodels{l}$ the set of all models of $\bisystem$, the corresponding set of sets of channels $\attack(\allmodels{l})$, obtained by point-wise application of $\attack$ to the elements of $\allmodels{l}$, contains under-approximations to all the attacks leading to $l$.

Hence, in order to solve the analysis we need essentially to compute all the models $\allmodels{l}$ of the propositional constraints $\bisystem$, that is, we have to solve the ALL-SAT problem for the input formula $\bisystem$. If no solution is found, i.e., $\bisystem$ is unsatisfiable, then the program point indicated by $l$ is not reachable. 

It is worthwhile observing that the translation into flow constraints over-approximates the behaviour of a process, giving rise to more executions than those actually arising in the semantics. Such spurious executions correspond to attacks that under-approximate the sets of channels required to reach the target $l$, and therefore the overall analysis results in an under-approximation. This intuition is formalised in the following correctness statement:
$$
\mbox{if}\quad P|Q \Longrightarrow^* C[{^lP'}]\quad \mbox{then}\quad \exists \mathcal{N} \in \attack(\allmodels{l})\ \mbox{s.t.}\ \mathcal{N} \subseteq \FC{Q}
$$
(where we slightly abuse the syntax writing $C[{^lP'}]$ for the sake of conciseness), i.e., for all the executions in which $Q$ drives $P$ to $l$, the analysis computes a set of channels $\mathcal{N} \in \Names$ that under-approximates the knowledge required of $Q$. The formulation of the actual theorem requires to establish some additional notation, and therefore is deferred to Appendix~\ref{sec:correctness}. In the following, we shall focus instead on an example process that pinpoints the imprecision of the analysis. 

A main source of over-approximation in the translation to flow constraints is the treatment of replication and restriction, whose interplay is simply disregarded by the analysis. As a matter of fact, a name restricted under replication is a different name in all the instances of the replicated process. Consider the following example:
$$
P \triangleq\ (\nu c)\left(   
\left(!\NEW{a}{^{\lab{1}}\IN{a}{x_a}.^{\lab{2}}\OUT{c}{c}}\right)
|
^{\lab{3}}\IN{c}{x_c}.^{\lab{4}}\IN{c}{x'_c}.^{\lab{5}}\dots
\right)
$$
Let label $\lab{5}$ be the location of interest and consider the following translation of $P$ into flow constraints, conveniently simplified for the sake of conciseness:
$$
(g_a \Leftrightarrow \overline{a})\  \land\  (g_c \lor \overline{a} \Leftrightarrow \overline{c})\  \land\  (\overline{c} \Leftrightarrow \overline{\lab{5}})\  \land\  (\ttt \Leftrightarrow \overline{\lab{5}})
$$
According to the translation, an attacker can reach label $\lab{5}$ either by knowing $a$ or by knowing $c$, that is, $\attack(\allmodels{\lab{5}}) = \{ \{a\}, \{c\}, \{a,c\}\}$.

Whilst it is clear that if the attacker guesses $c$ then the security checks at labels $\lab{3},\lab{4}$ can be satisfied, it is less obvious what it means for the attacker to guess $a$. In fact, if the adversary makes an output on $a$ twice, then two outputs on $c$ are triggered, and thus the checks on the path to the goal $\lab{5}$ are fulfilled. According to the semantics this may happen in all traces in which the replication is unfolded at least twice. Nonetheless, the two instances of $a$ in the two copies of the process are different: $\alpha$-renaming applies producing names $a$ and $a'$. Hence, by claiming that $\lab{5}$ can be reached by guessing $a$, the analysis under-approximates the knowledge required of the attacker, who needs $a$ and $a'$.
	
\subsection{Translating NemID}\label{sec:discovering-example}
The translation of the NemID system of \S~\ref{sec:example} returns the following flow constraints, augmented as explained above. For the sake of simplicity we omit the occurrences of $\ttt$ as a conjunct in all left-hand sides.
$$
\begin{array}{c}
\left.
\begin{array}{l}
\overline{\lab{1}}\\
\OSYN{cert} \Leftrightarrow \overline{x}_{\SYN{cert}}\\
\OSYN{id}	 \Leftrightarrow \overline{x}_{\SYN{id}}\\
\OSYN{pwd} \Leftrightarrow \overline{x}_{\SYN{pwd}}\\
\OSYN{otp} \Leftrightarrow \overline{x}_{\SYN{otp}}\\
\underbrace{\OSYN{cert} \lor (\OSYN{id} \land \OSYN{pwd} \land \OSYN{otp})}_{\varphi} \Leftrightarrow \overline{\lab{2}}\\
\varphi \land \overline{x}_{\SYN{cert}} \Leftrightarrow \overline{\lab{3}}\\
\varphi \land (\neg\overline{x}_{\SYN{cert}}) \Leftrightarrow \overline{\lab{4}}\\
\varphi \land (\neg\overline{x}_{\SYN{cert}}) \land \overline{x}_{\SYN{id}} \Leftrightarrow \overline{\lab{5}}\\
\varphi \land (\neg\overline{x}_{\SYN{cert}}) \land \overline{x}_{\SYN{id}} \land \overline{x}_{\SYN{pwd}} \land \overline{x}_{\SYN{otp}} \Leftrightarrow \overline{\lab{7}}\\
g_{\SYN{id}} \Leftrightarrow \OSYN{id}\\
g_{\SYN{pwd}} \Leftrightarrow \OSYN{pwd}\\
g_{\SYN{otp}} \Leftrightarrow \OSYN{otp}\\
g_{\SYN{cert}} \Leftrightarrow \OSYN{cert}
\end{array}
\right\}\begin{array}{c}\mbox{from }\\ \trprot{\it Applet}{\ttt}\end{array}
\end{array}
$$
$$
\begin{array}{ll}
\left.
\begin{array}{l}
\overline{\lab{8}}\\
\OSYN{id}	 \Leftrightarrow \overline{x'}_{\SYN{id}}\\
\OSYN{pin} \Leftrightarrow \overline{x}_{\SYN{pin}}\\
\OSYN{id} \land \OSYN{pin} \Leftrightarrow \overline{\lab{9}}\\
\OSYN{id} \land \OSYN{pin} \land \overline{x'}_{\SYN{id}} \Leftrightarrow \overline{\lab{10}}\\
\OSYN{id} \land \OSYN{pin} \land \overline{x'}_{\SYN{id}} \land \overline{x}_{\SYN{pin}} \Leftrightarrow \overline{\lab{11}}\\
g_{\SYN{pin}} \Leftrightarrow \OSYN{pin}
\end{array}
\right\}\begin{array}{c}\mbox{from }\\ \trprot{\it Mobile}{\ttt}\end{array}
\quad
&
\quad
\left.
\begin{array}{l}
\overline{\lab{12}}\\
\OSYN{login} \Leftrightarrow \overline{x}\\
\OSYN{login} \Leftrightarrow \overline{\lab{13}}\\
g_{\SYN{access}} \lor \OSYN{login} \Leftrightarrow \OSYN{access}
\end{array}\right\}\begin{array}{c}\mbox{from }\\ \trprot{\it Login}{\ttt}\end{array}
\end{array}
$$
\vspace{.3cm}
$$
\left(
g_{\SYN{login}} \lor 
\underbrace{\left(\varphi \land \overline{x}_{\SYN{cert}}\right)}_{\trprot{\it Applet}{\ttt}} \lor 
\underbrace{\left(\left(\varphi \land (\neg\overline{x}_{\SYN{cert}}) \land \overline{x}_{\SYN{id}} \land \overline{x}_{\SYN{pwd}} \land \overline{x}_{\SYN{otp}}\right)\right)}_{\trprot{\it Applet}{\ttt}} \lor 
\underbrace{\left(\left(\OSYN{id} \land \OSYN{pin} \land \overline{x'}_{\SYN{id}} \land \overline{x}_{\SYN{pin}}\right)\right)}_{\trprot{\it Mobile}{\ttt}}\right)
\vspace{.3cm}
\ \Leftrightarrow\ \OSYN{login}
$$
where the last formula combines the constraints that show the various ways to trigger an output on channel $\SYN{login}$. Notice how each formula models the checks on a given path: for being granted access, i.e., reaching label $\lab{13}$, a communicating process has to know channel $\SYN{login}$, as specified by the constraints derived from $\trprot{\it Login}{\ttt}$. In turn, the last formula describes what is needed in order to get hold of $\SYN{login}$, giving rise to a backward search procedure formalised in \S~\ref{sec:displaying}.

As for the solutions to the set of constraints, restricting to attacks we have:
$$
\begin{array}{lll}
\attack(\mu_{\mbox{\scriptsize triple}}) & = & \{\OSYN{id}, \OSYN{pwd}, \OSYN{otp}\}\\
\attack(\mu_{\mbox{\scriptsize cert}})   & = & \{\OSYN{cert}\}\\
\attack(\mu_{\mbox{\scriptsize mobile}}) & = & \{\OSYN{id}, \OSYN{pin}\}\\
\attack(\mu_{\mbox{\scriptsize login}}) & = & \{\OSYN{login}\}
\end{array}
$$
where the first two attacks come from the two ways of authenticating in the applet, the third from the mobile app, and the last one from guessing $\SYN{login}$, i.e., obtaining access in a way not encompassed by the system.

\subsection{Modularity and refinement}\label{sec:discovering-modularity}
Let us discuss now the role the analysis may take in a broader context, where it can be used to progressively refine the modelling of sub-systems which are revealed as potential attack targets, levering the modularity of process-algebraic specifications. 

Consider the NemID system and assume that a new way to access its services were offered by a service provider, which would authenticate a user via a phone number:
$$
\begin{array}{l}
{\it Phone} \triangleq\\\
^{\lab{14}}\IN{\SYN{phone}}{x_{\SYN{ph}}}.
\CASEL{\lab{15}}{x_{\SYN{ph}}}{y_{\SYN{ph}}}{^{\lab{16}}\OUT{\SYN{login}}{\SYN{ok}}}{\NIL}\\[2ex]
{\it NemID}' \triangleq (\nu\SYN{login})\dots(\nu\SYN{phone})\\
\quad\qquad\qquad (!{\it Login}\ |\ !{\it Applet}\ |\ !{\it Mobile}\ |\ !{\it Phone})
\end{array}
$$
Then we have  $\trprot{{\it NemID}'}{\ttt} = \trprot{\it NemID}{\ttt} \cup \trprot{\it Phone}{\ttt}$, that is, the translation of a new top-parallel process is independent from the formulae that have already been generated (but the quest for models has be re-computed with all the clauses). Obviously, due care has to be paid to names, e.g., the name $\SYN{login}$ in {\it Phone} has to be the same as the one used in process {\it NemID}. However, as we have seen, while restrictions play a crucial role in the semantics, they are simply ignored by the translation.

Besides being compositional with respect to the analysis of new components, the translation suitably integrates in a refinement cycle, where we start from a coarse abstraction of the system and then progressively refine those components that are revealed as candidates for being attacked, by replacing the corresponding set of formulae with a finer one. The constraint on names translates to a constraint on the interface of the component: if process $A$ is replaced by process $B$, then $B$ must be activated by the same inputs that activate $A$, and vice-versa, it must produce the same outputs towards the external environment that $A$ produces.

\section{Quantifying Attacks}\label{sec:quantifying}
The second challenge we tackle is quantifying attacks with respect to a general notion of cost. In the previous section we have presented a SAT-based solution technique, where each model of $\bisystem$ contains a set of channels that are necessary to fulfil an attack. Not all security mechanisms, however, offer the same protection guarantees, that is, not all channels are equal. A retinal scan can prove more difficult to bypass than a pin lock, and thus offer more protection. This is not the case, however, of an insider who is authorised to enter the corresponding room. A cost structure over channels facilitates formalising these considerations, and assigning costs to channels leads naturally to quantify sets of channels, that is, attacks.

The characterisation of attacks in terms of cost allows to order them, and ultimately to focus on those which are deemed the most likely given our understanding of the candidate attacker profiles. Moreover, the {\it minimal} cost of reaching a given location $l$ identifies the protection deployed to guard $l$ in the implementation, which can be contrasted with the desiderata of the specification. The higher the cost for the attacker, the higher the protection guarding the target.

In the following, we extend the qualitative analysis of \S~\ref{sec:discovering} to a quantitative setting. In particular, we should follow a modular approach according to which cost considerations are developed on top of the structure of the original analysis. The benefit of a layered strategy is two-fold: on the one hand, we present a technique that can be exploited to transform a great many qualitative analyses into quantitative analyses; on the other hand, whenever quantitative information about the entities in question is not available, we can resort to the qualitative solution.

\subsection{Security labels}\label{sec:quantifying-labels}
We have already seen how the calculus is instrumented with labels so as to refer to locations of interest. Such labels can be levered to build a {\it security lattice} $(\Sigma = \{\sigma_1,\dots,\sigma_n \}, \sqsubseteq_{\Sigma})$ specifying the desired protection deployed to guard each location. The set $\Sigma$ is equipped with the greatest lower bound operator $\sqcap_{\Sigma}$, and we assume to have a function $\security\! : \mathcal{L} \rightarrow \Sigma$ that maps labels into security levels. In particular, $\security(l_1) \sqsubseteq_{\Sigma} \security(l_2)$ denotes that the need for protection of the program point indicated by $l_2$ is greater than or equal to the need for protection of the program point indicated by $l_1$.

As an example of a security lattice, consider the \emph{military lattice} given by 
$$
\Sigma = \{\SYN{unclassified}, \SYN{confidential}, \SYN{secret}, \SYN{top}\mbox{-}\SYN{secret}\}
$$
with the ordering $\SYN{unclassified}$ $\sqsubset_{\Sigma}$ $\SYN{confidential} \sqsubset_{\Sigma}$ $\SYN{secret} \sqsubset_{\Sigma} \SYN{top}$-$\SYN{secret}$. More complex lattices, in particular non-linearly-ordered ones, are discussed in~\cite[Ch.~7]{Amoroso1994}.

In the NemID example of \S~\ref{sec:example}, we have observed that the system authenticates a user at label $\lab{13}$. In terms of security levels, we can rely on a simple security lattice ${\sf unrestricted} \sqsubset_{\Sigma} {\sf restricted}$, where $\security(\lab{13}) = {\sf restricted}$. Similarly, we would like to have the highest security level for labels $\lab{2}$ to $\lab{7}$ in process {\it Applet} and $\lab{9}$ to $\lab{11}$ in process {\it Mobile}. In fact, such labels are reached after fulfilling the security checks at labels $\lab{1}$ and $\lab{8}$, respectively, which have instead level ${\sf unrestricted}$, as label $\lab{12}$ in process {\it Login}.

\subsection{From qualitative to quantitative considerations}\label{sec:quantifying-costs}
Let $\cost$ be a function from channels $c \in \Names$ to costs $k \in \mathcal{K}$. Formally, we require $(\mathcal{K},\oplus)$ to be a commutative monoid (also known as Abelian monoid), that is, $\oplus$ is an associative and commutative binary operation on the set $\mathcal{K}$ and has an identity element. Moreover, we require $\mathcal{K}$ to be equipped with a partial order $\sqsubseteq_{\mathcal{K}}$, such that $(\mathcal{K}, \sqsubseteq_{\mathcal{K}})$ is a lattice, and $\oplus$ to be extensive, that is, the sum of two elements always dominates both the summands: 
\begin{equation}
\label{eq:extensive}
\forall k_1,k_2 \in \mathcal{K}\ .\ k_1 \sqsubseteq_{\mathcal{K}} (k_1 \oplus k_2)\ \land\ k_2 \sqsubseteq_{\mathcal{K}} (k_1 \oplus k_2)
\end{equation}
Finally, we assume that $\oplus$ is monotone and the least element $\bot \in \mathcal{K}$ is its identity element, that is, $\oplus$ is an upper bound operator of the lattice $(\mathcal{K}, \sqsubseteq_{\mathcal{K}})$, and therefore satisfies condition \eqref{eq:extensive}. For the sake of simplicity, we assume that the costs of channels are independent.

For the sake of simplifying the notation, in the following we shall feel free to apply the function $\cost$ to sets of names, according to the following definition:
$$
\begin{array}{l}
\cost: \mathcal{P}(\Names) \rightarrow \mathcal{K}\\
\cost(\{c_1,\dots,c_n\}) = \bigoplus_{i=1}^{n}\cost(c_i)
\end{array}
$$
Likewise, we extend the function $\cost$ also to sets of sets of names by point-wise application:
$$
\begin{array}{l}
\cost: \mathcal{P}(\mathcal{P}(\Names)) \rightarrow \mathcal{P}(\mathcal{K})\\[2ex]
\cost(\{c^1_1,\dots,c^1_{n_1}\},\dots,\{c^m_1,\dots,c^m_{n_m}\}) =\\[1ex]
\qquad  \big\{\cost(\{c^1_1,\dots,c^1_{n_1}\}), \dots, \cost(\{c^m_1,\dots,c^m_{n_m}\}) \big\}
\end{array}
$$
Therefore, given a model $\mu$ of $\bisystem$, the corresponding set $\attack(\mu) = \{c_1,\dots,c_n \}$ can be quantified as 
$$
\cost(\attack(\mu)) = \bigoplus_{c \in attack(\mu)} \cost(c)
$$
It is worthwhile noticing that this Boolean approach implies that multiple occurrences of the same basic action refer in fact to the same instance of the action, that is, to the same security mechanism, which once bypassed is bypassed ``forever''. Multiple instances of the same security mechanism must be represented with different names. This observation is intimately related to the over-approximating nature of the qualitative analysis discussed in \S~\ref{sec:discovering-solution}, where we have already mentioned that restrictions and replications are simply ignored.

As we mentioned above, however, a conservative approach to security would consider the attacks of {\it minimal} cost. Hence, given two attacks, i.e., two distinct models $\mu,\mu'$, we would discard $\mu'$ in case $\cost(\attack(\mu)) \sqsubset_{\mathcal{K}} \cost(\attack(\mu'))$. We can thus restrict the set of models $\allmodels{l}$ to the ones bearing attacks of minimal cost:
$$
\minimal(\allmodels{l}) = \left\{ \mu \in \allmodels{l}\, |\, \forall \mu' \in \allmodels{l} . \cost(\attack(\mu')) \not\sqsubset_{\mathcal{K}} \cost(\attack(\mu)) \right\}
$$
It is worthwhile noticing that $\minimal(\allmodels{l})$ may contain more than one model, as $(i)$ we consider all the attacks with same cost and $(ii)$ some attacks may have incomparable costs in case the cost set is not linearly ordered.

Now, we can relate an attack to the corresponding security level $\sigma \in \Sigma$ required to counter it by means of a function $\level\!: \mathcal{K} \rightarrow \Sigma$, compressing cost regions into security levels:
$$
\level(k) = 
\begin{cases}
\sigma_1 & \SYN{if}\ k \in \{k^1_1,\dots, k^1_{h_1}\}\\
\vdots\\
\sigma_m & \SYN{if}\ k \in \{k^m_1,\dots, k^m_{h_m}\}
\end{cases}
$$
where $\level$ is a well-defined function if the sets of costs $\{k^i_1,\dots, k^i_{h_i}\}$ are pair-wise disjoint and their union is $\mathcal{K}$. Moreover, it is natural to require that $\level$ is monotone. A simple example in the cost set $(\mathbb{N},+)$ and security lattice $\SYN{low} \sqsubset \SYN{medium} \sqsubset \SYN{high}$ is given by the choice
$$
\level(k) = 
\begin{cases}
\SYN{low} & 		\SYN{if}\; k \leq 1024\\
\SYN{medium} & 	\SYN{if}\; 1024 < k \leq 2048\\
\SYN{high} &		\SYN{if}\; 2048 < k
\end{cases}
$$
where numbers could represent the length of cryptographic keys, and we state for instance that a program point is poorly protected if no more than 1024 bits are necessary to attain it (for a fixed cryptosystem).

Finally, we extend $\level$ to work on sets of costs so as to encompass all the sets of channels produced by the analysis at once:
$$
\begin{array}{l}
\level: \mathcal{P}(\mathcal{K}) \rightarrow \mathcal{P}(\Sigma)\\
\level(\{k_1,\dots,k_n\}) = \{\level(k_1),\dots,\level(k_n) \}
\end{array}
$$
where the input $\{k_1,\dots,k_n\}$ is the set of costs of all minimal attacks, computed as
$$
\cost(\attack(\minimal(\allmodels{l})))
$$
Finally, the greatest lower bound $\sqcap_{\Sigma}$ is used to derive the greatest security level compatible with all attacks in $\minimal(\allmodels{l})$, that is, the protection of a program point corresponds at most to the cost of the weakest path leading to it.

%
%

A graphical illustration of the various components of the analysis is displayed in Fig.~\ref{fig:analysis}, where it is apparent how the quantitative analysis is built on top of the qualitative analysis. Intuitively, function $\security$ is the \emph{specification} expressing the target security architecture of a system with respect to a given security lattice, while $\level(\cost(\attack(\minimal(\allmodels{l}))))$ captures (an under-approximation of) how this architecture has been realised in the \emph{implementation}.

The overall aim of the analysis, i.e., checking whether the deployed protection lives up to the required security level, can thus be expressed by the property
$$
\forall l \in \mathcal{L}\ .\ \security(l) \sqsubseteq_{\Sigma} \level(\cost(\attack(\minimal(\allmodels{l}))))
$$
A violation of this condition is referred to as a potential \emph{inversion of protection}. The overall under-approximation of the analysis is the result of minimising over the costs of under-approximating sets of channels (cf. \S~\ref{sec:discovering-solution}).
\begin{figure}[t]
\centering
\begin{tikzpicture}[->,>=stealth',node distance=1.5cm]
  \node (n4) [] {$\mathcal{L}$};  
  \node (n0) [above=6.4cm of n4] {$\stackrel{?}{\sqsubseteq_{\Sigma}}$};    
  \node (nGhost) [above left of=n4] {};    
  \node (n3) [above right of=n4] {$\bisystem$};
  \node (n3top) [above of=n3] {$\allmodels{l}$};
  \node (n2) [above of=n3top] {$\mathcal{P}(\mathcal{K})$};
  \node (n1b) [above of=n2] {$\mathcal{P}(\Sigma)$};
  \node (n1a) [above of=n1b] {$\sigma'$};
  \node (n1)  [above=5.56cm of nGhost] {$\sigma$};
  \node [rotate=90] at (-4,3) (spec) {s p e c i f i c a t i o n};
  \node [rotate=90] at (5,3)  (imp) {i m p l e m e n t a t i o n};
\path (n4) 	edge[left] node {\footnotesize $\security$} (n1);
\path (n4)	edge[right,very near start] node {\footnotesize $\, \trprot{P}{\ttt}$} (n3);
\path (n3) 	edge[right] node {{\footnotesize ALL-SAT}} (n3top);
\path (n3top) 	edge[right] node {\footnotesize $\cost \circ \attack \circ \minimal$} (n2);
\path (n2) 	edge[right] node {\footnotesize $\level$} (n1b);
\path (n1b) 	edge[right] node {\footnotesize $\bigsqcap_{\Sigma}$} (n1a);
\end{tikzpicture}
\caption{The quantitative protection analysis at a glance for a fixed process $P$.}
\label{fig:analysis}
\end{figure}
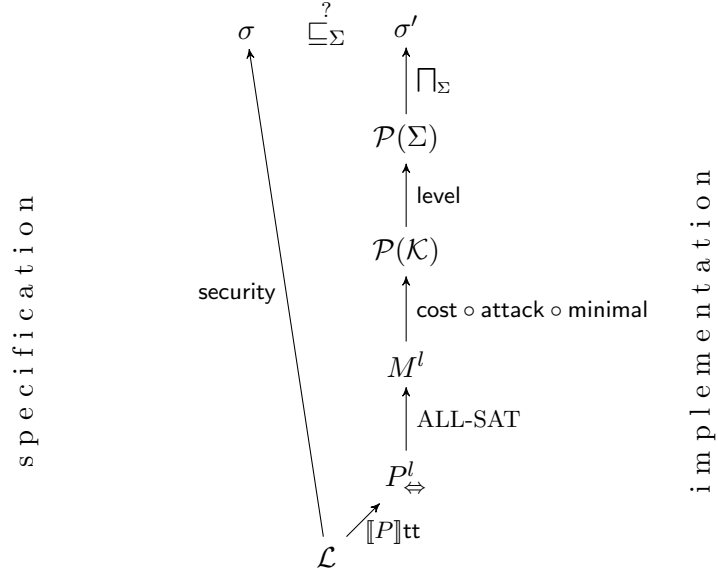

\subsection{Optmisation Modulo Theories}\label{sec:quantifying-solution}
In order to compute the set of sets of channels $\attack(\minimal(\allmodels{l}))$ that allow reaching $l$ incurring minimal costs, we need to solve an optimisation problem subject to the Boolean constraints $\bisystem$. There exist various techniques to cope with such problems, each suitable for particular choices of cost sets and objective functions. One solution is to first compute $\allmodels{l}$ and then minimise it by comparing models as explained above. Nonetheless, cost information can be levered to skip non-optimal models during the search, hence improving the performance. In the following, we show how to exploit an SMT solver to tackle the problem in its most general form. We limit to mention that linear programming techniques such as Pseudo-Boolean optimisation~\cite{Boros2002} are efficient alternatives for dealing with the monoid $(\mathbb{Z},+)$ and linear objective functions.

In a nutshell, our task reduces to compute models of $\bisystem$ containing attacks of minimal cost in the lattice $\mathcal{K}$. In other words, we are looking for {\it prime} implicants of $\overline{l}$~\cite{DeMicheli}, where primality is sought with respect to the given cost set.

Such an optimisation problem can be tackled by computing models for a list $\Pi_1,\dots,\Pi_n$ of SMT problems, where $\Pi_i$ is a more constrained version of $\Pi_{i-1}$ that requires to improve on the cost of the current solution. The initial problem $\Pi_1$ consists of the propositional constraints $\bisystem$ and of the objective function, whose value on the current model is stored in variable $\SYN{goal}$.

The objective function is essentially the cost of the current model. In order to compute $\cost(\attack(\mu))$ into variable $\SYN{goal}$ as part of $\mu$ itself we define:
$$
\SYN{goal} := \bigoplus_{i=1}^{n} ({\sf if}\ g_{c_i}\ {\sf then}\ \cost(c_i)\ {\sf else}\ \bot)
$$
where we combine the costs of all the channels that must be guessed, that is, the channels $c_i$'s such that the corresponding guessing literal $g_{c_i}$ is found to be $\ttt$. Otherwise, if a $g_{c_i}$ is $\fff$, then the corresponding $c_i$ needs not be guessed and its cost does not contribute to the cost of the attack. Recall that the least element $\bot$ of the cost lattice $\mathcal{K}$ does not contribute any cost, for it is the neutral element with respect to the cost combinator $\oplus$. Hence, by construction we have $\mu(\SYN{goal}) = \cost(\attack(\mu))$.

Then, while the problem is satisfiable, we improve on the cost of the current model by asserting new constraints which tighten the value of $\SYN{goal}$, until unsatisfiability is reported. Algorithm~\ref{alg:smt-loop} displays the pseudo-code of the procedure. In particular, observe that when a new problem $\Pi_i$ is generated, additional constraints are asserted that ask for $(i)$ a different model and $(ii)$ a non-greater cost: the former condition speeds up the search, while the latter explores the cost frontier.

The termination of the algorithm is ensured by the finiteness of possible models to the propositional variables of the $\Pi_i$'s, and by the fact that the same model cannot occur twice as solution due to the new constraints we generate in each iteration. At most, we need to solve as many $\Pi_i$'s as there are models of $\bisystem$, which coincide with the qualitative analysis (ALL-SAT). The correctness of the procedure stems from the fact that when unsatisfiability is claimed, by construction of the $\Pi_i$'s there cannot exist further models that comply with the cost constraints.
\begin{algorithm}[t]
 \KwData{The problem $\Pi \triangleq \bisystem \land \left( \SYN{goal} := \bigoplus_{i=1}^{n} ({\sf if}\ g_{c_i}\ {\sf then}\ \cost(c_i)\ {\sf else}\ \bot)\right)$}
 \KwResult{the set $\mathcal{M}$ of pairs $(\mu,\cost(\attack(\mu)))$ such that $\mu \in \minimal(\allmodels{l})$}
  $\mathcal{M} \leftarrow \emptyset$\;
 \While{$\Pi$ satisfiable}{
  $\mu$ $\leftarrow$ ${\sf get}$-${\sf model}(\Pi)$\;
  $k$ $\leftarrow$ $\mu(\SYN{goal})$\;
 \ForAll(\tcp*[f]{$\mu$ outperforms $\mu'$}){$(\mu',k') \in \mathcal{M}\ |\ k \sqsubset_{\mathcal{K}} k'$}{
       $\mathcal{M}$ $\leftarrow$ $\mathcal{M} \setminus \{(\mu',k')\}$
    }
  $\mathcal{M}$ $\leftarrow$ $\mathcal{M} \cup \{(\mu,k)\}$\;
  $\Pi$ $\leftarrow$ $\Pi \land \neg(\bigwedge_{i=1}^{n}(\overline{c_i} = \mu(\overline{c_i}))) \land \neg({\sf goal} \sqsupset_{\mathcal{K}} k)$\;
 }
\caption{The SMT-based solution procedure.}
\label{alg:smt-loop}
\end{algorithm}

It is worthwhile noticing how resorting to propositional logic integrates with the overall under-approximating nature of the analysis: a channel can either be learnt or not, and its cost contribute or not to the cost of an attack. This means that we do not keep track of the number of attempts made to guess some information, and always assume that guessing $c$ is successful whenever $g_c$ is found to be true. In order words, for a given cost set, we are considering the luckiest or cleverest attacker. As for the cost set (comparing and combining costs), SMT solvers offer native support for numeric costs and common mathematical functions, while more complex cost sets have to be encoded manually.

Finally, observe that the procedure above is not dependent on our analysis, but can be generally exploited to find optimal models of arbitrary logic formulae in arbitrary cost sets, and can be extended seamlessly to more complex logics.

\subsection{Attacking NemID}\label{sec:quantifying-example}
Consider the NemID system discussed in \S~\ref{sec:example}. There are several techniques for quantifying the cost of guessing secret information. {\it Quantification of information leakage}~\cite{Biondi2013} is an information theory-based approach for estimating the information an adversary gains about a given secret $s$ by observing the behaviour of a program parametrised on $s$. If $s$ is quantified in bits, then the corresponding information leaked by the program is quantified as the number of bits learnt by the adversary by observing one execution of the system. For instance, consider a test program $T$ parametrised on a secret password. $T$ inputs a string and answers whether or not the password is matched. Under the assumptions that the adversary knows the program and the length of the secret (no security-by-obscurity), we can estimate the knowledge gained by the adversary after one guessing attempt.

We leverage QUAIL~\cite{Biondi2013a}, a freely-available tool for quantifying information leakage, for determining costs to channels. Denoted ${\sf leak}(T,s)$ the leakage of $T$ on a secret $s$ as computed by QUAIL, we quantify the strength of a channel $c$ of $n$ bits as 
$$
\SYN{cost}(c) = \frac{n}{{\sf leak}(T,c)}
$$
where we assume the security offered by $c$ to be uniformly distributed over the $n$ bits. In this settings we are thus working in the cost monoid $(\mathbb{Q},+)$.

In our running example, the secrets to be guessed are $\SYN{pwd}$, $\SYN{otp}$, $\SYN{cert}$, $\SYN{pin}$, while we assume that $\SYN{id}$ is known to the attacker and thus has cost $0$ (in particular, in the NemID system is not difficult to retrieve such id, corresponding to the social security number of an individual). Moreover, we know that $\SYN{pwd}$ contains between 6 and 40 alphanumeric symbols and it is not case sensitive: assuming an average length of 10 symbols, given that there are 36 such symbols, we need $5.17$ bits to represent each symbol, for a total length of $52$ bits. Analogously, we determine the length of $\SYN{otp}$ as $20$ bits, while the length of the pin depends on the service provider: in case of a major bank it is just $14$ bits. As for the certificate, the authority is following NIST recommendations, using 2048-bit RSA keys for the time being, and for the sake of simplicity we assume that guessing an RSA key cannot be faster than guessing each of the bits individually. Finally, we disregard $\SYN{login},\SYN{access}$ by assigning them the least upper bound of the costs of all the other channels. Names used only as messages can be disregarded by assigning them cost $0$, as they do not influence an attack. Exploiting QUAIL and the formula defined above, we obtain the following $\cost$ map:
$$
\begin{array}{ll}
\cost(\SYN{pwd}) = 4.4\times 10^{15} & \cost(\SYN{pin})= 1.5\times 10^{4}\\
\cost(\SYN{otp}) = 10^6 & \cost(\SYN{cert}) = 3.4\times10^{616}
\end{array}
$$
Thus, limiting our attention to the set $\Names$ of channels occurring in the process {\it NemID}, the problem is to minimise
$$
\sum_{c \in \Names} \left({\sf if}\ g_c\ {\sf then}\ \cost(c)\ {\sf else}\ 0\right) 
$$
under the constraints given by $P_{\Leftrightarrow}^{\lab{13}}$. Instructed with this input, our procedure finds that the formula is satisfiable and the single cheapest model $\mu$ contains 
$$
g_\SYN{id}\mapsto \ttt\qquad g_\SYN{pwd}\mapsto \fff\qquad g_\SYN{mail}\mapsto \fff\qquad g_\SYN{pin}\mapsto \ttt\qquad g_\SYN{login}\mapsto \fff
$$
entailing $\attack(\minimal(\allmodels{\lab{13}})) = \{\{\SYN{id},\SYN{pin}\}\}$ and $\cost(\attack(\mu)) = \cost(\SYN{id}) + \cost(\SYN{pin}) = 1.5\times 10^{4}$.

As for the desired security levels, we observed in \S~\ref{sec:example} that $\security(\lab{13}) = \SYN{restricted}$ should hold. It is desirable to take as touchstone the protection offered by the applet, for it is standard among all service providers. The protection offered by the applet is the minimum between the cost of guessing a certificate and the cost of guessing the triple of credential, that is, $4.4\times 10^{15}+10^6$, therefore we should set
$$
\level(k) = 
\left\{
\begin{array}{ll}
\SYN{restricted}\ &  \SYN{if}\ k \geq 4.4\times 10^{15}+10^6\\
\SYN{unrestricted}\ &  \SYN{otherwise}
\end{array}
\right.
$$

We would like to verify that $\SYN{restricted} \sqsubseteq_{\Sigma} \level(\cost(\attack(\allmodels{\lab(13)})))$, which is false. Hence, it is the case that the implementation potentially guarantees less protection than the amount required by the specification, and thus we shall issue a warning to the designer of the system.

Finally, it is worthwhile noticing that the framework allows measuring the \emph{distance} between the implementation and the specification, and not only their qualitative compliance.

The analysis suggests that the most practicable way to break the authentication protocol is attacking the mobile app, as long as we believe that our cost map is sensible. For instance, a cryptographer would deem our assumptions on breaking RSA utterly unrealistic. In the following, we will discuss an effective alternative to the definition of numeric cost sets.

\subsection{Complex cost structures}\label{sec:quantifying-complex-costs}
So far we have worked with an example in the cost set $(\mathbb{Q},+)$, for it is natively encoded into SMT solvers and matches a first intuition of the notion of cost. Nonetheless, it is often difficult to provide an absolute estimate of the strength of a protection mechanism: sometimes different mechanisms are even incomparable, as cryptography and physical security might be. In such cases, it is more natural to describe the relative strength of a set of mechanisms with respect to each other. This is achieved by computing the analysis over symbolic and partially-ordered cost sets. Observe that Algorithm~\ref{alg:smt-loop} is already equipped to cope with the general problem of optimising on such cost sets.

As a basic example, consider the cost lattice displayed in Fig.~\ref{fig:cost-set}: we could characterise the cost of obtaining given information as $\SYN{cheap}$, if it does not require a specific effort, as $\SYN{cpu}$, if it requires significant computational capabilities (e.g., breaking an encryption scheme), as $\SYN{enrg}$, if it requires to spend a considerable amount of energy (e.g., engaging in the wireless exchange of a number of messages), or as $\SYN{expensive}$, if it requires both computations and energy. In order to combine such costs, a suitable choice is to take as monoid operator $\oplus$ the least upper bound $\sqcup$ of two elements in the cost lattice.

An interesting case of non-linear cost sets is offered by the study of security in Cyber-Physical Systems, where components combine both software and physical features~\cite{Vigo2012}. In particular, in such systems an attack could require to assemble cyber actions with physical tampering, whose costs can either be comparable or not depending on the nature of the quantities we are interested in (for instance, energy and memory are not directly comparable). 
\begin{figure}[t]
\centering
\begin{tikzpicture}[node distance=1.5cm]
	\node(zero) {{\sf cheap}};
	\node(one) [above left of=zero] {{\sf cpu}};
	\node(two) [above right of=zero] {{\sf enrg}};
	\node(up) [above left of=two] {{\sf expensive}};			

	\path (zero) edge (one)
   	    	(zero) edge (two);
	\path (one) edge (up);
	\path (two) edge (up);
\end{tikzpicture}
\caption{The Hasse diagram of a partially-ordered cost structure.}
\label{fig:cost-set}
\end{figure}
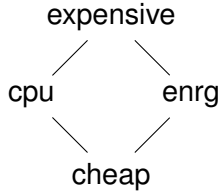

In conclusion, three elements push independently for the comprehensive SMT-based approach: the non-linearity of the cost set, its symbolic nature, and the non-linearity of the objective function.

\subsection{A semantic interpretation of guessing}\label{sec:quantifying-restriction}
It would be possible to formulate a neat semantic interpretation of the guessing capability of the attacker. Assume to deal with processes of the form $(\NEW{\overrightarrow{c}}{P})|Q$, where the first component is the system under study, in which all restrictions are at the outer-most level, and $Q$ is the attacker. Now, for $Q$ to interact with $P$, the attacker needs to move inside the scope of some restrictions so as to share some channel names with $P$. Whenever $Q$ enters the scope of a restriction $(\nu c)$, the name $c$ is guessed. In order to account for the cost $k \in \mathcal{K}$ of guessing a name, we can instrument each restriction with the corresponding cost, writing $(\nu^k c)$, and then augment the scope extension rule of Table~\ref{tab:congruence} so as to accumulate the cost of names that are guessed. The standard semantics of restriction used in security applications of process calculi, according to which a new name $c$ is only known to legal participants unless leaked ($P$, in our case), is encompassed by assigning $c$ an infinite cost.

Though possible, such an extension of the semantics is not necessary to prove the correctness of the analysis. Every assignment that satisfies the propositional constraints leads $Q$ to reach the location $l$ of interest, hence also the ones of minimal costs. Nonetheless, it is worthwhile noticing that such a quantitative point of view on restrictions generalises the distinction between the operators $\SYN{new}$ and $\SYN{hide}$ introduced in the \emph{secret} $\pi$-calculus~\cite{Giunti2012}. The operator $\SYN{hide}\, c$, which introduces a name $c$ inhibiting its scope extension, would correspond to $(\nu^{\infty} c)$, while we would have a more fine-grained view  on plain scope extension.

\section{Displaying Attacks}\label{sec:displaying}
We shall now embark in the last challenge of ours, that is, obtaining graphical representations of possible attacks that foster communicating effectively security information to non-experts, as these often are those in charge of taking decisions. We shall do this by means of attack trees, a widely-recognised tool for showing how a goal is attained in terms of combination of sub-goals.

Our developments on attack trees encompass both the qualitative analysis of \S~\ref{sec:discovering} and the quantitative extension of \S~\ref{sec:quantifying}. On the one hand, an attack tree contains {\it all} the attacks for a given target; on the other hand, a tree is construed as a propositional formula whose optimal models can be computed by means of the procedure of \S~\ref{sec:quantifying-solution}.

Attack trees are a widely-used graphical formalism for representing threat scenarios, as they appeal both to scientists, for it is possible to assign them a formal semantics, and to practitioners, for they convey their message in a concise and intuitive way (cf. references in \S~\ref{sec:related-trees}). In an attack tree, the root represents a target goal, while the leaves contain basic attacks whose further refinement is impossible or can be neglected. Internal nodes show how the sub-trees have to be combined in order to achieve the overall attack, and to this purpose propositional conjunction and disjunction are usually adopted as combinators. On top of this basic model, a number of extensions and applications of attack trees have been proposed, demonstrating how flexible and effective tool they are in practice. Figure~\ref{fig:tree-ex} displays a simplistic attack tree, where the overall goal of entering a bank vault is obtained by either bribing a guard or by stealing the combination and neutralising the alarm.
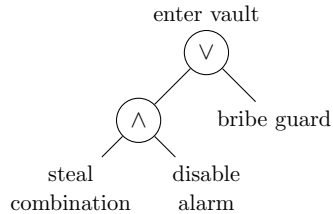
\begin{figure}[t]
\centering
\begin{tikzpicture}[node distance=1.5cm,scale=0.9,every node/.style={scale=0.85}] 
  \node[circle,draw,label=above:{\small enter vault}] (n1) [] {$\lor$};
  \node[circle,draw] (n2)  [below left of=n1] 	{$\land$};
  \node	 (n3)  [below right of=n1] 	{\small bribe guard};
  \node[align=center]	 (n2a) [below left of=n2] 		{\small steal\\ \small combination};
  \node[align=center]	 (n2b) [below right of=n2] 		{\small disable\\ \small alarm};
\path[-](n1) 	edge node {} (n2)
				edge node {} (n3)
		(n2) 	edge node {} (n2a)
				edge node {} (n2b);
\end{tikzpicture}
\caption{How to enter a bank vault, for dummies.}
\label{fig:tree-ex}
\end{figure}

\subsection{Synthesising attack trees}\label{sec:displaying-inferring}
In the following, we shall rely on the translation $\trprot{P}{\ttt}$ devised in Table~\ref{tab:translation-clauses} but embrace a slightly different interpretation so as to allow explicitly generating trees. In particular, we replace bi-implications in $\bisystem$ with implications and we ignore literals $g$, obtaining a set of constraints denoted $\rightsystem$. In order to obtain attack trees as commonly defined in the literature, in the following we assume that all the quality guards $q$ in the process under study are linear.

Given a process $P$ and a label $l$ occurring in $P$, we generate a formula $\DENOTATION{l}$ representing the attacks reaching $l$ by backward chaining the formulae in $\rightsystem$ so as to derive $\overline{l}$. It is central to observe that the procedure re-establishes the original system of bi-implications thus guaranteeing the correctness of the analysis in terms of compatibility with the developments of \S~\ref{sec:discovering}.

Before explaining the algorithm, it is worthwhile discussing the nature of the backward chaining-like procedure defined in the following. Standard backward chaining~\cite[Ch.~7]{Russell2009} combines Horn clauses so as to check whether a given goal follows from the knowledge base. Instead, we are in fact trying to derive all the knowledge bases that allow inferring the goal given the inference rules $\rightsystem$, which are not strict Horn clauses as they can contain more than one positive literal. The backward-chaining point of view stresses the relationship of our problem to the quest for implicants of $\overline{l}$, as we have already observed.

The rules for generating $\DENOTATION{l}$ are displayed in Table~\ref{tab:translation-trees}. For our formulae are propositional, there is no unification other than syntactical identity of literals involved in the procedure. Notice that the algorithm only applies valid inference rules.

Rule (Sel) selects the antecedent of the formula leading to the goal $\overline{l}$: since there is a unique such rule, in order to derive $\overline{l}$ we have to derive the antecedent $\varphi$ of $\varphi \Rightarrow \overline{l}$. Observe that we are not interested in deriving $\overline{l}$ in any other way: for $\overline{l}$ is derived assuming $\varphi$, the original bi-implication format  $\varphi \Leftrightarrow \overline{l}$ is re-established.

Rule (Pone-c) encodes either a tautology (if $\overline{c}$ has to be inferred then $\overline{c}$ is in the knowledge base) or applications of modus ponens ($\overline{c}$ is derived assuming $\varphi$, thanks to $\varphi \Rightarrow \overline{c}$): the whole rule is an instance of disjunction introduction.

This is the point where our algorithm differs from plain backward chaining: since we are building the knowledge bases that allow inferring $\overline{l}$, whenever we encounter a literal $\overline{c}$ we need to account for all the ways of deriving $\overline{c}$, namely by placing $\overline{c}$ itself in the knowledge base or by satisfying a rule whose consequent is $\overline{c}$. Moreover, $\overline{c}$ plays now the role of $g_c$: for the procedure assumes $c$ without further deriving it, we can safely re-use the literal, whose semantics coincides now with the one of $g_c$. Observe that we could consider literals $g$ and their relation to channels explicitly, but this would impact the size of the tree, hence its readability.

Similarly, rule (Pone-x) encodes an application of modus ponens, taking advantage of the uniqueness of $\varphi \Rightarrow \overline{x}$ (cf. Lemma~\ref{lm:close-x}).

Rules (Tolle-) collect applications of modus tollens (law of contrapositive) in the classic backward fashion (i.e., when considering the derivation from the leaves to the root such steps would encode that modus). Rules (DM-) encode De Morgan's laws. Finally, rules (Comp-) simply state the compositionality of the procedure.

It is worthwhile observing that in classic backward chaining, loops are avoided by checking whether a new sub-goal (i.e., a literal to be derived) is already on the goal stack (i.e., is currently being derived). Component $\mathcal{D}$ in Table~\ref{tab:translation-trees} is in charge of keeping track of the current goals, but this is done on a local basis as opposed to the traditional global stack, that would result if $\mathcal{D}$ were treated as a global variable. As shown in \S~\ref{sec:displaying-remark}, in our setting the global stopping criterion would lead to unsound results. Moreover, observe that using the local environment $\mathcal{D}$ we lose the linear complexity in $|\rightsystem|$ typical of backward chaining, and incur an exponential complexity in the worst case. Nonetheless, observe that this theoretical bound is not incurred systematically, but it depends on the shape of the process under study. Finally, a local stack allows to re-use channel literals in rule (Pone-c).

Notice that we do not need to keep track of literals $\overline{x}$ in $\mathcal{D}$, as we cannot meet with a cycle because a variable cannot be used prior to its definition (in virtue of Lemma~\ref{lm:close-x}).

Finally, observe that a parse tree $T_l$ of $\DENOTATION{l}$ is an attack tree, showing how $l$ can be reached by combining the knowledge of given channels. The internal nodes of the tree contain a Boolean operator in $\{\land,\lor\}$, while the leaves contain literals representing the knowledge of channels. As De Morgan's laws are used to push negations to literals of $\DENOTATION{l}$, negation can only occur in the leaves of $T_l$. In the following, we shall manipulate attack trees always at their denotation level, that is, the object under evaluation is $\DENOTATION{l}$ as opposed to its representation $T_l$ (cf. \S~\ref{sec:implementation-discussion}).

For the sake of discussion, it is worthwhile noticing that the procedure for generating $\DENOTATION{l}$ can be used to generate a tree explicitly during the computation, or even an {\sc And-Or} graph~\cite[Ch.~4]{Russell2009}. It is unclear to us, however, whether the more compact graph representation would be simpler for non-expert to understand.
\begin{table}[t]
\small
\caption{Synthesising the propositional formula $\DENOTATION{l}$ for the attack tree $T_l$.}
\hrule
$$
\begin{array}{lr}
\DENOTATION{l} = \trtree{\varphi}{\emptyset}\qquad \mbox{\sf\footnotesize where}\ (\varphi \Rightarrow \overline{l}) \in \rightsystem \qquad\qquad \mbox{\small (Sel)}
\end{array}
$$
\vspace{-.3cm}
\textcolor{gray}{\hrule}
$$
\begin{array}{l}
\begin{array}{ll}
\trtree{\overline{c}}{\mathcal{D}} = \overline{c}\ \lor 
\left\{
\begin{array}{ll}
\trtree{\varphi}{(\mathcal{D} \cup \{\overline{c}\})} & \mbox{\sf\footnotesize if}\ \overline{c}\not\in \mathcal{D}, { \sf\footnotesize where}\  (\varphi \Rightarrow \overline{c}) \in \rightsystem\\
\fff  & \mbox{\sf\footnotesize otherwise}
\end{array}
\right.\qquad
\hfill \mbox{\small (Pone-c)}\\[4ex]
\trtree{\neg \overline{c}}{\mathcal{D}}	= 
\left\{
\begin{array}{ll}
\trtree{\neg\varphi}{(\mathcal{D} \cup \{\neg\overline{c}\})} & \mbox{\sf\footnotesize if}\ \neg\overline{c}\not\in \mathcal{D}, { \sf\footnotesize where}\  (\varphi \Rightarrow \overline{c}) \in \rightsystem\\
\ttt  & \mbox{\sf\footnotesize otherwise}
\end{array}
\right.\qquad
\hfill \mbox{\small (Tolle-c)}\\[4ex]
\trtree{\overline{x}}{\mathcal{D}} = \trtree{\varphi}{\mathcal{D}}\qquad\quad\ \mbox{\sf\footnotesize where}\ (\varphi \Rightarrow \overline{x}) \in \rightsystem \hfill \mbox{\small (Pone-x)}\\[2ex]
\trtree{\neg \overline{x}}{\mathcal{D}} = \trtree{\neg \varphi}{\mathcal{D}}\qquad \mbox{\sf\footnotesize where}\ (\varphi \Rightarrow \overline{x}) \in \rightsystem \hfill \mbox{\small (Tolle-x)}\\[2ex]
\end{array}\\
\begin{array}{l}
\trtree{\neg(\varphi_1 \land \dots \land \varphi_n)}{\mathcal{D}} = \trtree{\neg \varphi_1}{\mathcal{D}} \lor \dots \lor \trtree{\neg \varphi_n}{\mathcal{D}}\qquad\qquad\qquad\  \hfill \mbox{\small (DM-1)}\\
\trtree{\neg(\varphi_1 \lor \dots \lor \varphi_n)}{\mathcal{D}} = \trtree{\neg \varphi_1}{\mathcal{D}} \land \dots \land \trtree{\neg \varphi_n}{\mathcal{D}}\ \hfill \mbox{\small (DM-2)}\\[1ex]
\trtree{\varphi_1 \land \dots \land \varphi_n}{\mathcal{D}} = \trtree{\varphi_1}{\mathcal{D}} \land \dots \land \trtree{\varphi_n}{\mathcal{D}}\ \hfill \mbox{\small (Comp-1)}\\
\trtree{\varphi_1 \lor \dots \lor \varphi_n}{\mathcal{D}} = \trtree{\varphi_1}{\mathcal{D}} \lor \dots \lor \trtree{\varphi_n}{\mathcal{D}}\ \hfill \mbox{\small (Comp-2)}\\[2ex]
\qquad\qquad\trtree{\ttt}{\mathcal{D}} = \ttt\qquad\qquad\qquad \trtree{\fff}{\mathcal{D}} = \fff
\end{array}
\end{array}
$$
\hrule
\label{tab:translation-trees}
\end{table}

Finally, observe that $\DENOTATION{l}$ only contains literals $\overline{c}$ corresponding to channels, that is, the backward chaining-like procedure described above and formalised in Table~\ref{tab:translation-trees} eliminates all the literals $\overline{x}$. Therefore, the reachability of $l$ is only expressed in terms of knowledge of channels. This result is formalised in Lemma~\ref{lm:ground-l}, and guarantees that a map from channels to cost suffices to quantify an attack.

\subsection{The attack tree for NemID}\label{sec:displaying-example}
Consider the process \emph{NemID} discussed in \S~\ref{sec:example} and its translation ${\it NemID}^{\lab{13}}_{\Rightarrow}$. Figure~\ref{fig:tree-tool} shows the attack tree $T_{\lab{13}}$, as generated by our implementation, presented in \S~\ref{sec:implementation-tool}. The denotation of $T_{\lab{13}}$ is given by the following formula:
$$
\begin{array}{lll}
\DENOTATION{\lab{13}} & = & \OSYN{login}\ \lor\\
& & \left( (\OSYN{cert} \lor (\OSYN{id} \land \OSYN{pwd} \land \OSYN{otp})) \land \OSYN{cert}\right)\ \lor\\
& & \big( (\OSYN{cert} \lor (\OSYN{id} \land \OSYN{pwd} \land \OSYN{otp})) \land (\neg \OSYN{cert}) \land\ \OSYN{id} \land \OSYN{pwd} \land \OSYN{otp}\big) \ \lor\\
& & \left( \OSYN{id} \land \OSYN{pin}\right)
\end{array}
$$
As a matter of fact, the algorithm tends to generate simple but redundant formulae, which can be simplified automatically, e.g., via a reduction to a normal form. The following formula, for instance, is equivalent to $\DENOTATION{\lab{13}}$ but highlights more clearly the ways in which an attack can be carried out:
$$
\OSYN{login}\ \lor\ (\overline{\SYN{id}} \land \overline{\SYN{pin}})\ \lor\ 
(\overline{\SYN{id}} \land \overline{\SYN{pwd}} \land \overline{\SYN{otp}})\ \lor\
\overline{\SYN{cert}}
$$
Observe that the formula above is in Disjunctive Normal Form (DNF). Such normal form has the merit of providing an immediate intuition of the alternative conditions that lead to reach the program point under study, as displayed in Fig.~\ref{fig:tree-dnf}. However, the conversion to DNF may cause an exponential blow-up in the number of literals, and compact translations require to introduce fresh atoms, garbling the relation between the tree and the original system. Therefore, we did not implement such conversion in the tool.

Finally, notice that the disjunct $\OSYN{login}$ encodes the possibility of obtaining a login token in any other way not foreseen in the system, and thus accounts for all the attacks not explicitly related to the shape of our formalisation. Such possibility can be rules out by assigning the maximum possible cost to the channel, as we have seen in \S~\ref{sec:quantifying-example},
\begin{figure*}[t]
\centering
\subfigure[$T_{\lab{13}}$ as displayed by the Quality Protection Tool, presented in \S~\ref{sec:implementation-tool}.]{
\label{fig:tree-tool}
\includegraphics[scale=.38]{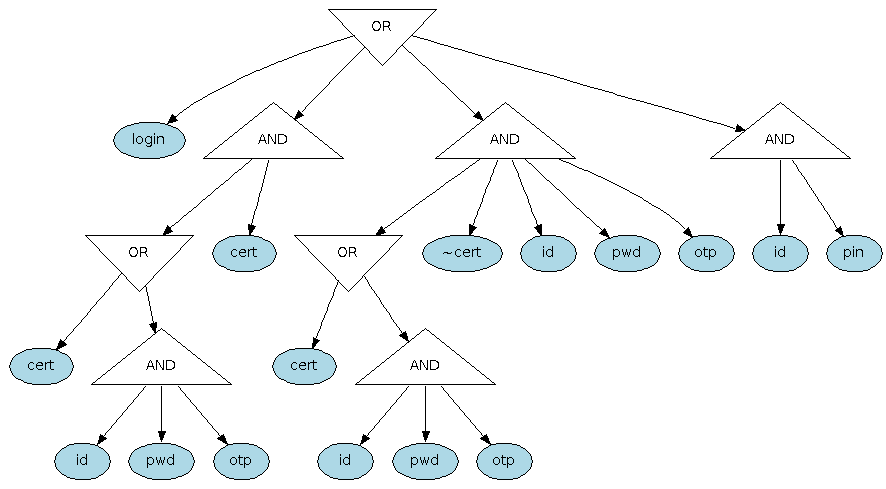}
}
\subfigure[The simplified DNF attack tree.]{
\label{fig:tree-dnf}
\begin{tikzpicture}[node distance=1.5cm,scale=0.9,every node/.style={scale=0.85}] 
  \node[circle,draw,label=above:{$T_{\lab{13}}$}] (n1) [] {$\lor$};
  \node[circle,draw] (n2) [below left=1cm and 1.2cm of n1] {$\land$};
  \node				 (n0) [below left=.9cm and -.1cm of n1] {$\OSYN{login}$};
  \node				 (n3) [below right=.9cm and -.1cm of n1] 		{$\OSYN{cert}$};
  \node[circle,draw] (n4) [below right=1cm and 1.2cm of n1] {$\land$};
  \node				 (n2a) [below left=.5cm and -.1cm of n2] 		{$\OSYN{id}$};
  \node				 (n2b) [below right=.5cm and -.1cm of n2] 		{$\OSYN{pin}$};
  \node				 (n5) [below left of=n4] 		{$\OSYN{id}$};
  \node				 (n6) [below right of=n4] 		{$\OSYN{pwd}$};
  \node				 (n7) [below=.45cm of n4] 		{$\OSYN{otp}$};
\path[-](n1) 	edge node {} (n2)
				edge node {} (n0)
				edge node {} (n3)
			 	edge node {} (n4)
		(n2) 	edge node {} (n2a)
				edge node {} (n2b)
		(n4) 	edge node {} (n5)
				edge node {} (n6)
				edge node {} (n7);
\end{tikzpicture}
}
\caption{The attack tree $T_{\lab{13}}$ of the NemID example.}
\label{fig:tree}
\end{figure*}
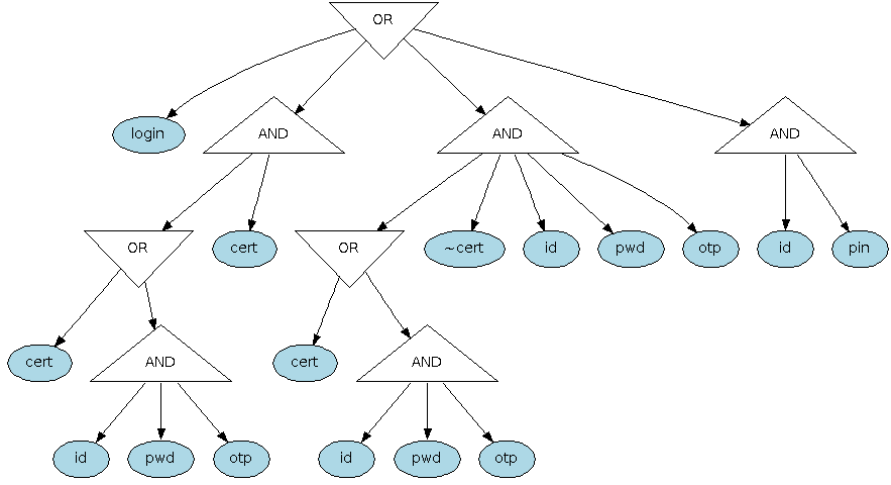
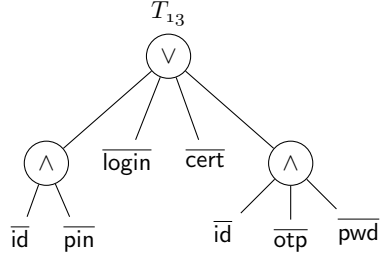

\subsection{Global stopping criterion}\label{sec:displaying-remark}
We conclude this section by showing why the global stopping criterion is unsound for the procedure of Table~\ref{tab:translation-trees}. Consider the following set of formulae:
$$
\overline{a} \Rightarrow \overline{b}\qquad \overline{b} \Rightarrow \overline{a}\qquad \overline{a} \land \overline{b} \Rightarrow \overline{\lab{7}}
$$
which stems from a conveniently simplified translation of the process
$$
P \triangleq\ ^{\lab{1}}\IN{a}{x_a}.^{\lab{2}}\OUT{b}{b}\, |\, ^{\lab{3}}\IN{b}{x_b}.^{\lab{4}}\OUT{a}{a}\, |\, 
^{\lab{5}}\IN{a}{x'_a}.^{\lab{6}}\IN{b}{x'_b}.^{\lab{7}}\OUT{c}{c}
$$
The generation of $\DENOTATION{\lab{7}}$ unfolds as follows:
$$
\DENOTATION{\lab{7}} = \trtree{\overline{a} \land \overline{b}}{\emptyset} = \trtree{\overline{a}}{\emptyset} \land \trtree{\overline{b}}{\emptyset} 
$$
where, in particular, it is
$$
\begin{array}{l}
\trtree{\overline{a}}{\emptyset} = \overline{a} \lor \trtree{\overline{b}}{\{\overline{a}\}} = \overline{a} \lor \overline{b} \lor \trtree{\overline{a}}{\{\overline{a},\overline{b}\}} = \overline{a} \lor \overline{b} \lor \fff = \overline{a} \lor \overline{b}\\
\trtree{\overline{b}}{\emptyset} = \overline{b} \lor \trtree{\overline{a}}{\{\overline{b}\}} = \overline{b} \lor \overline{a} \lor \trtree{\overline{b}}{\{\overline{b},\overline{a}\}} = \overline{b} \lor \overline{a} \lor \fff = \overline{b} \lor \overline{a}\\
\end{array}
$$
leading to $\DENOTATION{\lab{7}} = \overline{a} \lor \overline{b}$, which is consistent with the reachability of label $\lab{7}$ in $P$.

Assume now to carry out the generation of $\DENOTATION{\lab{7}}$ applying a global stopping criterion, that is, to keep track of derived goals in a global environment, initially empty. We would obtain:
$$
\trtree{\overline{a}}{\emptyset} = \overline{a} \lor \trtree{\overline{b}}{\{\overline{a}\}} = \overline{a} \lor \overline{b} \lor \trtree{\overline{a}}{\{\overline{a},\overline{b}\}} = \overline{a} \lor \overline{b} \lor \fff = \overline{a} \lor \overline{b}\\
$$
at this point, however, the environment contains $\overline{a},\overline{b}$, and thus the generation of $\trtree{\overline{b}}{}$ leads to $\overline{b}$, resulting in $\DENOTATION{\lab{7}} = (\overline{a} \lor \overline{b}) \land \overline{b}$, which is not satisfied by the model where only $\overline{a}$ is $\ttt$, and thus is wrong. Analogously, we would obtain a wrong result if we chose to unfold $\trtree{\overline{b}}{}$ before $\trtree{\overline{a}}{}$.

\newpage
\section{Implementation of the Analysis}\label{sec:implementation}

\subsection{Discussion}\label{sec:implementation-discussion}
As the backward-chaining procedure on $\rightsystem$ re-establishes bi-implications and only applies valid inference rules, $\DENOTATION{l}$ and $\bisystem$ are equisatisfiable. They are not equivalent, i.e., in general their models do not coincide, as $\DENOTATION{l}$ only contains channel literals, but they contain the same {\it attacks}. Hence, we can solve the quantitative version of the protection analysis in either way:
\begin{itemize}
\item generate $\bisystem$, compute the models bearing minimal attacks, and extract the corresponding attacks; or,
\item generate $\rightsystem$, derive $\DENOTATION{l}$, and compute its minimal models,
\end{itemize}
finally relating the result to the security lattice by means of the function $\level$.

It is worthwhile observing that while the procedure for generating trees is exponential in the worst case, the size of $\DENOTATION{l}$ is much smaller than the size of $\bisystem$, and therefore it is not necessarily the case that the overall running time would increase when undertaking the tree generation. Though our example set is not extensive enough for supporting any final claim, still it is interesting to comment briefly how the analyses on $\bisystem$ and on $\DENOTATION{l}$ behave in terms of running time.

Consider the NemID system. The translation to $\bisystem$ takes about one fourth of the time the computation of $\DENOTATION{l}$ takes (in the order of seconds). Solving the optimisation problem on $\bisystem$ takes about 1.25 the time it takes on $\DENOTATION{l}$. Nonetheless, the second step is much more demanding in terms of performance, so that on average the two approaches take the same amount of time. The same applies to the example discussed in~\cite{Vigo2014}. Increasing the size of the process under study it seems that the generation of attack trees, while exponential in general, tends to outperform the overall analysis on $\bisystem$.

It is worthwhile noticing that comparing the two approaches reduces to establishing whether it is faster to find models of $\bisystem$ or $\DENOTATION{l}$, for the translation time is negligible as the size of processes increases. Even limiting to the core propositional structure of the problem, there is no conclusive answer to the question, whose investigation falls outside the scope of this work. For an introduction to the problem of efficiency of satisfiability the reader is referred to~\cite[Ch.~9,13]{Lipton2009} and~\cite[\S~7.6.3]{Russell2009}; research work on the subject is for instance in~\cite{Achlioptas2009,Coja-Oghla2013}. For performance of SMT solvers refer to \url{http://smtcomp.sourceforge.net/}.

Finally, let us remark that our evaluation of a tree $T_l$ is, by definition, the evaluation of $\DENOTATION{l}$. In other words, the minimal cost of a tree is the minimal cost of a model $\mu$ of \DENOTATION{l} - we evaluate \DENOTATION{l}, as opposed to evaluating $T_l$. Another approach would be to evaluate the tree itself by means of traversing its structure, which however would lead to  different results, unless a multiset-based model is adopted, such as the one proposed by Mauw and Oostdik~\cite{Mauw2006}. A more precise treatment of the Boolean approach to analyzing attack trees is presented in~\cite{Aslanyan2015}.

\subsection{The Quality Protection Tool}\label{sec:implementation-tool}
A proof-of-concept implementation of the framework has been developed in Java and is available at
{\small
\begin{center}
\url{http://www.imm.dtu.dk/~rvig/quality-protection.html}
\end{center}
}
\noindent together with the code for the NemID example described in the text.

The \emph{Quality Protection Tool} takes as input an ASCII representation of a Value-Passing Quality Calculus process $P$ and a label $l$, and generates the flow constraints $\rightsystem$. 

Furthermore, the tool implements the backward-chaining procedure of \S~\ref{sec:displaying-inferring}, relying on our own simple infrastructure for propositional logic, as available libraries tend to avoid the explicit representation of implications, that is instead handy in our case during the backward-chaining computation with non-Horn-like clauses. Once the backward-chaining procedure is executed, and thus $\DENOTATION{l}$ has been derived, the tool can graphically represent the corresponding tree $T_l$, thanks to an encoding in DOT\footnote{\url{http://www.graphviz.org/}} and using ZGRViewer\footnote{\url{http://zvtm.sourceforge.net/zgrviewer.html}} for displaying the tree.

Finally, we have implemented the optimisation loop on top of the Z3 SMT solver\footnote{\url{http://z3.codeplex.com/}} (Java API). The tool resorts to Z3 for numerical cost sets, optimising the sum of the costs. As for symbolic and non-linearly-ordered cost sets, a finite lattice can be fed into the tool, and the least upper bound is used as monoid operator. Costs can be specified in two ways: numeric costs can be directly fed to the tool, while before specifying symbolic costs the finite lattice $(\mathcal{K},\sqsubseteq)$ has to be loaded. In order to specify a lattice, one has to declare $\top$ and $\bot$, and then the operator $\oplus$ as a list of entries $x\oplus y=z$. The names of the elements of the lattice and the partial order $\sqsubseteq$ are automatically inferred from the graph of $\oplus$. 

All these components are glued together thanks to a simple graphical interface. The implementation of the analysis based on $\bisystem$~\cite{Vigo2014} is also available and this, together with improvements on the translation from processes to clauses, is the main difference with respect to the tool of~\cite{Vigo2014}.

\section{Related Work}\label{sec:related}

\subsection{Protection analysis}\label{sec:related-analysis}
Our work is inspired by a successful strand of literature in protocol verification, where a protocol is translated into a set of first-order Horn clauses and resolution-based theorem proving is used to establish security properties~\cite{Paulson1998,Weidenbach1999}, and by the flow logic approach to static analysis~\cite{Nielson2012a}. In particular, the translation from processes to propositional formulae is inspired by ProVerif \cite{Blanchet2009a} translation of protocols into first-order Horn clauses, but can be more formally understood as a flow logic where the carrier logic is not the usual Alternation-free Least Fixed Point Logic, since it cannot express optimisation problems. Moreover, the main problem we discuss is propositional; ideas for a first-order extension are sketched in \S~\ref{sec:fo-trees}.

In order to formalise the ``need for protection'' of a location we resort to security lattices, that are widely used for describing levels of security in access control policies. An excellent introductory reference is~\cite[Chs.~6,7]{Amoroso1994}.

As for the solution technique we exploit, different approaches have been presented to solve optimisation problems via SMT. In particular, Nieuwenhuis and Oliveras~\cite{Nieuwenhuis2006} proposed to modify the DPLL$(T)$ procedure inherent to SMT solvers so as to look for optimal assignments, while Cimatti et al.~\cite{Cimatti2010} developed the search for an optimal assignments on top of an SMT solver, as we do in \S~\ref{sec:quantifying-solution}. Nonetheless, both these works focus on numeric weights, which in our settings are represented with linearly ordered cost structures. Our more general notion of weight is modelled after Meadows's cost sets, formalised as monoids in~\cite{Meadows2001}.

Finally, another perspective on the technical developments underpinning the analysis points to computing \emph{prime implicants} of a given formula~\cite{Dillig2012}, where in our case primality is sought with respect to the cost set.

\subsection{Attack trees}\label{sec:related-trees}
Graphical representations of security threats are often used to convey complex information in an intuitive way. Formalisation of such graphical objects are referred to chiefly as \emph{attack graphs}~\cite{Phillips1998,Jha2002,Sheyner2004,Mehta2006} and \emph{attack trees}~\cite{Schneier1999,Sheyner2002,Mauw2006,Rehak2009,Jurgenson2010}. In this work we prefer the phrase ``attack trees'', but our procedure can be adapted to generate attack graphs. We refer the reader to \cite{Kordy20141} for a recent survey on the vast literature about attack trees, while in the following we retrace some of historical developments on modelling, generating, and analysing attack trees that inspired our developments.

While different authors have different views on the information that should decorate such objects, instrumental to the analysis that the tree or the graph is supporting, all definitions share the ultimate objective of showing how atomic attacks (i.e., the leaves) can be combined to attain a target goal (i.e., the root). This perspective is enhanced in the seminal work of Schneier~\cite{Schneier1999}, that found a great many extensions and applications. In particular, Mauw and Oostdijk~\cite{Mauw2006} lay down formal foundations for attack trees, while Kordy et al.~\cite{Kordy} and Roy et al.~\cite{Roy2012} suggest ways to unify attacks and countermeasures in a single view. Even though Schneier's work is mostly credited for having introduced attack trees, and it had certainly a crucial role in making attack trees mainstream in computer security, the origin of this formalism can be traced back to fault trees, expert systems (e.g., Kuang~\cite{Baldwin1987}), and privilege graphs~\cite{Dacier1996}.

As for the automated generation of attack graphs, the literature is skewed towards the investigation of network-related vulnerabilities: available tools expect as input rich models, including information such as the topology of the network and the set of atomic attacks to be considered. The backward search techniques of Phillips and Swiler~\cite{Phillips1998} and Sheyner et al.~\cite{Sheyner2002} have proven useful to cope with the explosion of the state space due to such expressive models. However, the search has to be carried out on a state space that is exponential in the number of system variables, whose construction is the real bottle-neck of these approaches, and the result graph tends to be large even if compact BDD-based representations are used, as argued in~\cite{Ammann2002}. In particular, in~\cite{Sheyner2002} a model checking-based approach is developed, where attack graphs are characterised as counter-examples to safety properties; a detailed example is discussed in~\cite{Sheyner2004}. Similarly to Phillips and Swiler, we adopt an attacker-centric perspective, which cannot simulate benign system events such as the failure of a component, as in~\cite{Sheyner2002}. Directly addressing the exponential blow-up of~\cite{Sheyner2002}, Ammann et al.~\cite{Ammann2002} propose a polynomial algorithm, but the drop in complexity relies on the assumption of monotonicity of the attacker actions and on the absence of negation. On the same line, Ou et al.~\cite{Ou2006} present an algorithm which is quadratic in the number of machines in the network under study.

As for the analyses developed on top of attack trees, we present a reachability analysis which computes the cheapest sets of atomic attacks that allow attaining a location of interest in the system, as it is standard in the attack tree literature. This approach seamlessly encompasses the probabilistic analysis of~\cite{Sheyner2002,Sheyner2004} (costs to atomic attacks would represent their likelihood and the objective function would compute the overall probability) and offers a uniform framework to address other quantitative questions~\cite{Bistarelli2007,Kordy2012}. The NP-completeness of our SMT-based approach is in line with the complexity of the minimisation analysis of~\cite{Sheyner2002,Sheyner2004}.

Finally, Mehta et al.~\cite{Mehta2006} present a technique for ranking sub-graphs so as to draw attention to the most promising security flaws. Whilst we do not directly tackle this issue, for condensing an entire tree into a formula we gain in performance but we lose the original structure, a post-processing step could be undertaken to compute the value of the internal nodes (sub-formulae).

\section{Conclusion}\label{sec:conclusion}
Discovering attacks is an essential part of investigating security. Quantifying the attacks is necessary when dealing with complex systems and facing budget considerations. Both tasks risk to become a fruitless exercise if their findings cannot be communicated effectively to decision-makers.

Static analysis of process-algebraic specifications offers a unifying framework where these three challenges can be addressed uniformly and by means of modular developments. Our approach can be exploited in a great many context to temper qualitative verification methods with quantitative considerations.

To support this claim, we have developed a protection analysis over the Value-Passing Quality Calculus where attacks can be automatically inferred, quantified, and displayed. At the heart of the analysis lies the abstraction of security mechanisms with secure channels, which allows to define security checks as input actions and thus attacks as sets of channels over which communication must take place. Moreover, channels lend themselves naturally to support cost considerations, as they represent different security mechanisms.

Starting from this basic assumption, we have developed a qualitative analysis to discover all sets of channels in terms of models of a logic representation of the system under study, resorting to propositional satisfiability (SAT). Levering costs to channels, we have then enhanced the analysis with a quantitative layer, enriching the original SAT problem with the notion of cost of a model and developing an SMT-based optimisation procedure for computing optimal models. Finally, by means of a backward-chaining search on the constraints representing a system, we have shown how to infer an attack tree for a given target, leading to another characterisation of the quantitative analysis, possibly cheaper to compute.

Our SMT-based optimisation allows reasoning with symbolic and non-linearly ordered cost structures, as it is often more natural to describe the relationships between different protection mechanisms instead of assigning them absolute numbers. What is more, this technique is exploitable in all the contexts where models of a formula have to be ranked according to given criteria.

A step necessary to exploit fully the expressiveness of our SMT approach to optimisation is the refactoring of the implementation, so as to produce a stand-alone version of the solution engine, which takes as input an SMT problem (problem, cost lattice, objective function) and returns its optimal models. This would allow to compare our technique to those surveyed above and in particular with the forthcoming optimising version of Z3.

On the modelling side, we present in Appendix~\ref{sec:fo-trees} how to lift the developments to the full Quality Calculus. Nonetheless, it is unclear to us whether the resulting attack trees would benefit their intended users, for the additional information may reduce readability drastically. It seems instead promising to investigate further the notion of priced restriction discussed in \S~\ref{sec:quantifying-restriction}, and to compare its expressiveness to other approaches recently presented in the literature. Finally, as highlighted by Meadows~\cite{Meadows2001}, the monoid of the cost set needs not be commutative, as the order in which costs are paid might influence their combination. It would be interesting to investigate mechanisms for re-determining costs dynamically, as a process is evaluated.

\section*{Acknowledgement}
 Special thanks to Zaruhi Aslanyan and Alessandro Bruni for many inspiring and fruitful discussions.

\bibliographystyle{plain}
\bibliography{library}

\begin{thebibliography}{10}

\bibitem{Achlioptas2009}
Dimitris Achlioptas.
\newblock {Random Satisfiability}.
\newblock In {\em Handbook of Satisfiability}, volume 185, pages 245--270. IOS
  Press, 2009.

\bibitem{Ammann2002}
Paul Ammann, Duminda Wijesekera, and Saket Kaushik.
\newblock {Scalable, graph-based network vulnerability analysis}.
\newblock In {\em 9th ACM conference on Computer and Communications Security
  (CCS'02)}, pages 217--224, 2002.

\bibitem{Amoroso1994}
Edward Amoroso.
\newblock {\em {Fundamentals of Computer Security Technology}}.
\newblock Prentice-Hall, 1994.

\bibitem{Aslanyan2015}
Zaruhi Aslanyan and Flemming Nielson.
\newblock Pareto efficient solutions of attack-defence trees.
\newblock In {\em Principles of Security and Trust - 4th International
  Conference, {POST} 2015, Held as Part of the European Joint Conferences on
  Theory and Practice of Software, {ETAPS} 2015, London, UK, April 11-18, 2015,
  Proceedings}, LNCS, pages 95--114. Springer, 2015.

\bibitem{Baldwin1987}
Robert~W. Baldwin.
\newblock {\em {Rule Based Analysis of Computer Security}}.
\newblock PhD thesis, MIT, 1987.

\bibitem{Biondi2013}
Fabrizio Biondi, Axel Legay, Pasquale Malacaria, and Andrzej Wąsowski.
\newblock {Quantifying Information Leakage of Randomized Protocols}.
\newblock In {\em 14th International Conference Verification, Model Checking,
  and Abstract Interpretation (VMCAI'13)}, volume 7737 of {\em LNCS}, pages
  68--87. Springer, 2013.

\bibitem{Biondi2013a}
Fabrizio Biondi, Axel Legay, and Louis-marie Traonouez.
\newblock {QUAIL : A Quantitative Security Analyzer}.
\newblock In {\em 25th International Conference on Computer Aided Verification
  (CAV'13)}, volume 8044 of {\em LNCS}, pages 702--707. Springer, 2013.

\bibitem{Bistarelli2007}
Stefano Bistarelli, Marco Dall'Aglio, and Pamela Peretti.
\newblock {Strategic games on defense trees}.
\newblock In {\em Formal Aspects in Security and Trust (FAST'06)}, number~3 in
  LNCS, pages 1--15. Springer, 2007.

\bibitem{Blanchet2009a}
Bruno Blanchet.
\newblock {Automatic verification of correspondences for security protocols}.
\newblock {\em Journal of Computer Security}, 17(4):363--434, 2009.

\bibitem{Boros2002}
Endre Boros and Peter~L Hammer.
\newblock {Pseudo-boolean optimization}.
\newblock {\em Discrete Applied Mathematics}, 123(1-3):155--225, 2002.

\bibitem{Cimatti2010}
Alessandro Cimatti, Anders Franz\'{e}n, Alberto Griggio, Roberto Sebastiani,
  and Cristian Stenico.
\newblock {Satisfiability Modulo the Theory of Costs: Foundations and
  Applications}.
\newblock In {\em Tools and Algorithms for the Construction and Analysis of
  Systems}, volume 6015 of {\em LNCS}, pages 99--113, 2010.

\bibitem{Coja-Oghla2013}
Amin Coja-Oghla and Konstantinos Panagiotou.
\newblock {Going after the K-SAT Threshold}.
\newblock In {\em 45th ACM symposium on Theory of Computing (STOC'13)}, pages
  705--714. ACM, 2013.

\bibitem{Dacier1996}
M.~Dacier, Y.~Deswarte, and M.~Kaaniche.
\newblock {Models and tools for quantitative assessment of operational
  security}.
\newblock In {\em 12th International Information Security Conference
  (IFIP/SEC'96)}, pages 177--186, 1996.

\bibitem{DeMicheli}
Giovanni De~Micheli.
\newblock {\em Synthesis and Optimization of Digital Circuits}.
\newblock McGraw-Hill, 1994.

\bibitem{Dillig2012}
Isil Dillig, Thomas Dillig, Kenneth~L. McMillan, and Alex Aiken.
\newblock {Minimum Satisfying Assignments for SMT}.
\newblock In {\em Computer Aided Verification (CAV'12)}, volume 7358 of {\em
  LNCS}, pages 394--409. Springer, 2012.

\bibitem{Giunti2012}
Marco Giunti, Catuscia Palamidessi, and Frank~D. Valencia.
\newblock {Hide and New in the Pi-Calculus}.
\newblock In {\em Proceedings Combined 19th International Workshop on
  Expressiveness in Concurrency and 9th Workshop on Structured Operational
  Semantics (EXPRESS/SOS)}, volume~89, pages 65--79, 2012.

\bibitem{Jha2002}
Somesh Jha, Oleg Sheyner, and Jeannette~M. Wing.
\newblock {Two formal analyses of attack graphs}.
\newblock In {\em Proceedings 15th IEEE Computer Security Foundations Workshop
  CSFW15}, pages 49--63, 2002.

\bibitem{Jurgenson2010}
Aivo J\"{u}rgenson and Jan Willemson.
\newblock {Serial Model for Attack Tree Computations}.
\newblock In {\em Information, Security and Cryptology (ICISC'09)}, volume 5984
  of {\em LNCS}, pages 118--128. Springer, 2010.

\bibitem{Kordy}
Barbara Kordy, Sjouke Mauw, Sasa Radomirovic, and Patrick Schweitzer.
\newblock {Foundations of Attacks-Defense Trees}.
\newblock In {\em 7th International Workshop on Formal Aspects of Security and
  Trust (FAST'10)}, volume 6561 of {\em LNCS}, pages 80--95. Springer, 2010.

\bibitem{Kordy2012}
Barbara Kordy, Sjouke Mauw, and Patrick Schweitzer.
\newblock {Quantitative Questions on Attack-Defense Trees}.
\newblock In {\em 15th International Conference on Information Security and
  Cryptology (ICISC'12)}, volume 7839 of {\em LNCS}, pages 49--64. Springer,
  2012.

\bibitem{Kordy20141}
Barbara Kordy, Ludovic Piètre-Cambacédès, and Patrick Schweitzer.
\newblock Dag-based attack and defense modeling: Don’t miss the forest for
  the attack trees.
\newblock {\em Computer Science Review}, 13–14:1 -- 38, 2014.

\bibitem{Lipton2009}
R.~J. Lipton.
\newblock {\em {The P=NP Question and G\"{o}del's Lost Letter}}.
\newblock Springer, 2009.

\bibitem{Mauw2006}
Sjouke Mauw and Martijn Oostdijk.
\newblock {Foundations of Attack Trees}.
\newblock In {\em 8th International Conference on Information Security and
  Cryptology (ICISC'05)}, volume 3935 of {\em LNCS}, pages 186--198. Springer,
  2006.

\bibitem{Meadows2001}
Catherine Meadows.
\newblock {A cost-based framework for analysis of denial of service in
  networks}.
\newblock {\em Journal of Computer Security}, 9(1):143--164, 2001.

\bibitem{Mehta2006}
Vaibhav Mehta, Constantinos Bartzis, Haifeng Zhu, Edmund Clarke, and Jeannette
  Wing.
\newblock {Ranking Attack Graphs}.
\newblock In {\em 9th International Symposium on Recent Advances in Intrusion
  Detection (RAID'06)}, volume 4219 of {\em LNCS}, pages 127--144, 2006.

\bibitem{StandardML}
R.~Milner, M.~Tofte, R.~Harper, and D.~MacQueen.
\newblock {\em {The Definition of Standard ML (Revised)}}.
\newblock MIT Press, 1997.

\bibitem{Modersheim2009}
Sebastian M\"{o}dersheim and Luca Vigan\`{o}.
\newblock {Secure Pseudonymous Channels}.
\newblock In {\em 14th European Symposium on Research in Computer Security
  (ESORICS'09)}, volume 5789 of {\em LNCS}, pages 337--354. Springer, 2009.

\bibitem{Nielson2002HardestAttaker}
Flemming Nielson, Hanne {Riis Nielson}, and Ren\'{e}~Rydhof Hansen.
\newblock {Validating firewalls using flow logics}.
\newblock {\em Theoretical Computer Science}, 283(2):381--418, 2002.

\bibitem{Nieuwenhuis2006}
Robert Nieuwenhuis and Albert Oliveras.
\newblock {On SAT Modulo Theories and Optimization Problems}.
\newblock In {\em Theory and Applications of Satisfiability Testing (SAT'06)},
  volume 4121 of {\em LNCS}, pages 156--169, 2006.

\bibitem{Ou2006}
Xinming Ou, Wayne~F. Boyer, and Miles~a. McQueen.
\newblock {A scalable approach to attack graph generation}.
\newblock In {\em Proceedings of the 13th ACM conference on Computer and
  communications security - CCS '06}, page 336, New York, New York, USA, 2006.
  ACM Press.

\bibitem{Paulson1998}
Lawrence~C Paulson.
\newblock {The inductive approach to verifying cryptographic protocols}.
\newblock {\em Journal of Computer Security}, 6(1-2):85--128, 1998.

\bibitem{Phillips1998}
Cynthia Phillips and Laura~Painton Swiler.
\newblock {A graph-based system for network-vulnerability analysis}.
\newblock In {\em Proceedings of the 1998 workshop on New security paradigms
  NSPW 98}, volume pages, pages 71--79, 1998.

\bibitem{Rehak2009}
Martin Reh\'{a}k, Eugen Staab, Volker Fusenig, Michal P\v{e}chou\v{c}ek, Martin
  Grill, Jan Stiborek, Karel Barto\v{s}, and Thomas Engel.
\newblock {Runtime Monitoring and Dynamic Reconfiguration for Intrusion
  Detection Systems}.
\newblock In {\em Recent Advances in Intrusion Detection (RAID'09)}, volume
  5758 of {\em LNCS}, pages 61--80. Springer, 2009.

\bibitem{Nielson2012a}
Hanne {Riis Nielson}, Flemming Nielson, and Henrik Pilegaard.
\newblock {Flow Logic for Process Calculi}.
\newblock {\em ACM Computing Surveys}, 44(1):1--39, January 2012.

\bibitem{Nielson2012}
Hanne {Riis Nielson}, Flemming Nielson, and Roberto Vigo.
\newblock {A Calculus for Quality}.
\newblock In {\em 9th International Symposium on Formal Aspects of Component
  Software (FACS'12)}, volume 7684 of {\em LNCS}, pages 188--204. Springer,
  2012.

\bibitem{Roy2012}
Arpan Roy, Dong~Seong Kim, and Kishor~S. Trivedi.
\newblock {Attack countermeasure trees (ACT): towards unifying the constructs
  of attack and defense trees}.
\newblock {\em Security and Communication Networks}, 5(8):929--943, 2012.

\bibitem{Russell2009}
Stuart Russell and Peter Norvig.
\newblock {\em {Artificial Intelligence: A Modern Approach}}.
\newblock Prentice-Hall, 3rd edition, 2009.

\bibitem{Schneier1999}
Bruce Schneier.
\newblock {Attack Trees}.
\newblock {\em Dr. Dobb's Journal}, 1999.

\bibitem{Sheyner2002}
Oleg Sheyner, Joshua~W. Haines, Somesh Jha, Richard Lippmann, and Jeannette~M.
  Wing.
\newblock {Automated Generation and Analysis of Attack Graphs}.
\newblock In {\em 2002 IEEE Symposium on Security and Privacy}, pages 273--284,
  2002.

\bibitem{Sheyner2004}
Oleg Sheyner and Jeannette~M. Wing.
\newblock {Tools for Generating and Analyzing Attack Graphs}.
\newblock In {\em 2nd International Symposium on Formal Methods for Components
  and Objects (FMCO'03)}, volume 3188 of {\em LNCS}, pages 344--371, 2004.

\bibitem{Vigo2012}
Roberto Vigo.
\newblock {The Cyber-Physical Attacker}.
\newblock In {\em 7th ERCIM/EWICS Workshop on Cyberphysical Systems}, volume
  7613 of {\em LNCS}, pages 347--356. Springer, 2012.

\bibitem{Vigo2013}
Roberto Vigo, Flemming Nielson, and Hanne {Riis Nielson}.
\newblock {Broadcast, Denial-of-Service, and Secure Communication}.
\newblock In {\em 10th International Conference on integrated Formal Methods
  (iFM'13)}, volume 7940 of {\em LNCS}, pages 410--427, 2013.

\bibitem{Vigo2014a}
Roberto Vigo, Flemming Nielson, and Hanne {Riis Nielson}.
\newblock {Automated Generation of Attack Trees}.
\newblock In {\em 27th Computer Security Foundations Symposium (CSF'14)}, pages
  337--350. IEEE, 2014.

\bibitem{Vigo2014}
Roberto Vigo, Flemming Nielson, and Hanne {Riis Nielson}.
\newblock {Uniform Protection for Multi-exposed Targets}.
\newblock In {\em 34th IFIP International Conference on Formal Techniques for
  Distributed Objects, Components and Systems (FORTE'14)}, volume 8461 of {\em
  LNCS}, pages 182--198. Springer, 2014.

\bibitem{Weidenbach1999}
Christoph Weidenbach.
\newblock {Towards an automatic analysis of security protocols in first-order
  logic}.
\newblock In {\em 16th International Conference on Automated Deduction (CADE-
  16)}, pages 314--328. Springer-Verlag, 1999.

\end{thebibliography}

\appendix

\section{Correctness of the Protection Analysis}\label{sec:correctness}
The correctness of the protection analysis with respect to the semantics of the calculus is formalised as follows:
$$
\mbox{if}\quad P|Q \Longrightarrow^* C[{^lP'}]\quad \mbox{then}\quad \exists \mathcal{N} \in \attack(\allmodels{l})\ \mbox{s.t.}\ \mathcal{N} \subseteq \FC{Q}
$$
i.e., for all the executions in which $Q$ drives $P$ to $l$, the analysis computes a set of channels $\mathcal{N} \in \Names$ that under-approximates the knowledge required of $Q$.

Technically, it is convenient to organise a formal proof in two steps. First, if $P|Q$ reaches $l$ then $P|H[\FC{Q}]$ reaches $l$, where process $H$ is the hardest attacker possible and is parametrised on the knowledge of $Q$. $H$ can be thought as the (infinite) process executing all possible actions on $\FC{Q}$, and the proof simply argues that whatever $Q$ can, $H$ can (Fact~\ref{fact:hardest}). A similar approach is detailed in~\cite{Nielson2002HardestAttaker}.

Finally, the second step shows that if $P|H[\mathcal{N}']$ reaches $l$, then there must be a set $\mathcal{N} \in \attack(\allmodels{l})$ such that $\mathcal{N} \subseteq \mathcal{N}'$ (Theorem~\ref{th:correctness}). Observe that this formulation corresponds to the qualitative analysis of \S~\ref{sec:discovering}. However, as $\attack(\minimal(\allmodels{l})) \subseteq \attack(\allmodels{l})$, the correctness of the quantitative analysis follows as a particular case.

\begin{definition}[Hardest attacker]
Let $\mathcal{N} = \{c_1, \dots, c_n\}$ be a finite set of channels. $H[\mathcal{N}]$ is the process that does all possible sequence of output actions over channels in $\mathcal{N}$:
$$
H[\mathcal{N}] \triangleq \NEW{d}{\left(!(\OUT{c_1}{d})\right)|\dots|\left(!(\OUT{c_n}{d})\right)}
$$
(where labels are of no use hence omitted).
\end{definition}

In the definition we used a fresh name $d$ as output term, but any name can be chosen as $P$ cannot check the content of input variables. Observe that the channels in $\mathcal{N}$ might be used to trigger necessary outputs on other channels, according to the constraints in $\bisystem$.

\begin{fact}\label{fact:hardest}
Let $P,P',Q$ be processes and $C,C'$ contexts. It holds that
$$
\mbox{if}\quad P|Q\, \Longrightarrow^* C[{^lP'}]\quad \mbox{then}\quad  P|H[\FC{Q}]\, \Longrightarrow^* C'[{^lP'}]
$$
\end{fact}

As a matter of fact, the only blocking actions in $P$ are inputs, and since the calculus is value-passing, the execution of $P$ is driven exclusively by $(i)$ the number of output actions $Q$ performs, $(ii)$ the channels over which they are executed, and $(iii)$ their order. Now, for each channel $c \in \FC{Q}$, that is, for each channel known to $Q$, by construction $H[\FC{Q}]$ interleaves an arbitrary number of output on $c$, thus mimicking all the possible sequence of output actions on $\FC{Q}$, among which is the one performed by $Q$.

The main correctness result is phrased as follows. Since the semantics is value-passing, the proof does not present any particular obstacle, and therefore we limit to present its structure and major cases.

\begin{theorem}[Correctness of the protection analysis]\label{th:correctness}
Let $P,P'$ be processes, $C$ a context, and $\mathcal{N} \in \Names$ a set of channels. It holds that
$$
\mbox{if}\quad P|H[\mathcal{N}] \Longrightarrow^* C[{^lP'}]\quad \mbox{then}\quad  \left(\exists \mathcal{N}'\, .\, \mathcal{N}' \in \attack(\allmodels{l})\, \land\, \mathcal{N}' \subseteq \mathcal{N}\right)
$$
\end{theorem}

{\it Proof sketch.} 
The proof is organised by induction on the length $k$ of the derivation sequence $P|H[\mathcal{N}] \Longrightarrow^* C[{^lP'}]$.\\

{\it Basis.} If $k=0$, then it is $P = C[{^lP'}]$, from which $\emptyset \in \attack(\allmodels{l})$, for $\overline{l}$ is a fact in $\bisystem$, and thus it does not entail any channel literal to be $\ttt$. Since $\emptyset \subseteq \mathcal{N}$, for all set $\mathcal{N}$, the thesis follows.\\

{\it Step.} Assume $k=k_0+1$. The derivation sequence can be written as 
$$
P|H[\mathcal{N}] \Longrightarrow^{k_0} C''[^{l'}P''] \Longrightarrow C'[{^lP'}]
$$
for some context $C''$ and process $P''$. The inductive hypothesis applies to the first $k_0$ steps of the derivation: there exists $\mathcal{N}'' \in \attack(\allmodels{l'})$ such that $\mathcal{N}'' \subseteq \mathcal{N}$. Now, it suffices to show that the last step in the derivation sequence, leading to reaching $l$, preserves the inclusion relationship. 

The last reduction $C''[P''] \Longrightarrow C'[{^lP'}]$ is a short-hand writing that conflates a number of cases, but observe that it must be entailed by combining rule (Sys) with a transition $P'' \strans{\lambda} P'$, where we assume that the contexts $C'', C'$ take care of hiding restrictions preceding $P''$ and parallel components of $P'', P'$ that are not affected by the transition. As for the congruence step in the premise of rule (Sys), observe that the rewrite cannot produce inputs or outputs not already considered by the analysis, as the latter always assumes replications to be unfolded. To conclude, a formal proof requires an induction on the shape of the inference tree for the transition $P'' \strans{\lambda} P'$.

Let us comment upon the case of rule (In-tt), which is the most interesting and the only non-trivial. Assume that the binder $b$ is a simple input $\IN{c}{x}$, passing which the label of interest is attained. Now, it is either $c \in \mathcal{N}''$, in which case we conclude $\mathcal{N}' = \mathcal{N}'' \subseteq \mathcal{N}$, or $c \notin \mathcal{N}''$. Again, we have two cases.

If there exists a subset of $\mathcal{N}''$ which can trigger another component of $P$ to make an output on $c$, we again conclude $\mathcal{N}' = \mathcal{N}'' \subseteq \mathcal{N}$. Otherwise, we are in the case $\varphi \land \overline{c} \Leftrightarrow \overline{l}$ with $c \notin \mathcal{N}''$ and $\overline{c}$ not a consequence of the literals corresponding to $\mathcal{N}''$. In the set of constraints we have $g_c \lor \varphi' \Leftrightarrow \overline{c}$. It must then be either $c \in \mathcal{N}$, or $\mathcal{N}''' \subseteq \mathcal{N}$, where $\mathcal{N}'''$ satisfies $\varphi'$, otherwise $\mathcal{H}[N]$ would not pass the input. If we look at models of $\bisystem$, we have that either $g_c$ is $\ttt$ or $\varphi'$ evaluates to $\ttt$ -- in every model. In the first case we conclude $\mathcal{N}' = \mathcal{N}'' \cup \{c\} \subseteq \mathcal{N}$. In the latter $\mathcal{N}' = \mathcal{N}'' \cup \mathcal{N}^{iv} \subseteq \mathcal{N}$, with $\mathcal{N}^{iv} \subseteq \mathcal{N}'''$, because the least way of satisfying $\varphi'$ by the analysis under-approximates the least way of satisfying $\varphi'$ by the semantics.

The same reasoning applies to the case in which $b$ is a quality binder, as formulae are computed according to the semantics of quality binders.

\section{Properties of Attack Trees}\label{sec:technical-trees}
This appendix contains some results that substantiate the procedure for generating attack trees discussed in \S~\ref{sec:displaying-inferring}.

\begin{lemma}\label{lm:close-x}
Let $P$ be a process. For any variable $x$ in $P$, there exists exactly one formula $\varphi \Rightarrow \overline{x}$ in the translation $\rightsystem$, and $\overline{x}$ does not occur in $\varphi$.
\end{lemma}

\begin{proof}
By induction on the structure of processes. In particular, observe that in Table~\ref{tab:translation-clauses} a literal $\overline{x}$ is added to $\varphi$ only when a $\SYN{case}$ clause is met, and by hypothesis $\overline{x}$ must previously appear in a binder, for processes are closed. Finally, recall that we assume processes to be renamed apart (cf- \S~\ref{sec:calculus-syntax}), hence the same variable or name cannot be bound twice.
\end{proof}

Let us discuss now the \emph{complexity} of the translation given in Table~\ref{tab:translation-clauses}. Let $\size(\mathcal{C})$ denote the number of literals occurring in a set of formulae $\mathcal{C}$, that is, $\size(\mathcal{C}) = \sum_{\varphi \in \mathcal{C}}\left( \size(\varphi)\right)$, where $\size(\varphi)$ counts the literals in $\varphi$.

\begin{lemma}\label{lm:complexity-1}
Let $P$ be a process, and assuming that $P$ contains $n$ actions. Then $\size(\trprot{P}{\ttt}) = O(n^2)$.
\end{lemma}

\begin{proof}
If $P$ consists of $n$ actions, $\trprot{P}{\ttt}$ consists of at most $O(n)$ formulae. More in detail, $\trprot{P}{\ttt}$ consists of $n_o + n_c + n_i + n_b$ formulae, $n_o$ being the number of outputs in $P$, $n_c$ the number of $\SYN{case}$ clauses, $n_i$ the number of simple inputs (including the ones occurring within quality binders), and $n_b$ the number of binders. The number of literals in a formula depends linearly on the number of actions preceding the label at which the formula is generated (cf. Table~\ref{tab:translation-clauses}), hence the number of literals in $\trprot{P}{\ttt}$ is asymptotically bounded by $n^2$.
\end{proof}

It is interesting to observe that from a theoretical point of view $O(n^2)$ is a precise bound to $\size(\trprot{P}{\ttt})$. Consider the process ${\it IN}_n$ that consists of $n$ sequential inputs $\IN{c_1}{x_1}.\dots .\IN{c_n}{x_n}$. The number of literals in $\trprot{{\it IN}_n}{\ttt}$ grows with 
$$
\begin{array}{lll}
\sum_{i=1}^n \left(2(i-1)+3\right) & = & \sum_{i=1}^n \left( 2i+1\right) =\\
								  & = & n + 2\sum_{i=1}^n i = n + 2\frac{n(n+1)}{2} =\\
								  & = & n^2 + 2n
\end{array}
$$
where $i$ records the number of literals in the hypothesis $\varphi$, we have omitted counting the $\ttt$ conjuncts, and we leverage the fact that an input generates two formulae whose size is $\size(\varphi)+2$ adding $1$ literal to the hypothesis, from which the relation $2(i-1) + 3$ is derived. Similarly, the translation of a process made of alternating inputs and $\SYN{case}$ clauses would grow quadratically (with greater constants than ${\it IN}_n$).

\begin{lemma}\label{lm:ground-l}
Let $P$ be a process. For all labels $l$ occurring in $P$, the formula $\DENOTATION{l}$ built according to the rules of Table~\ref{tab:translation-trees} contains no literal $\overline{x}$. In particular, $\DENOTATION{l}$ only contains literals related to channels $c$. 
\end{lemma}

\begin{proof}
By induction on the number of steps in the unfolding of the generation of $\DENOTATION{l}$, according to the rules in Table~\ref{tab:translation-trees}.
\end{proof}

\section{First-Order Attack Trees}\label{sec:fo-trees}
We present in this section an extension to the framework whose detailed development deserves to be deepened in future work. The ideas discussed in the following have not been implemented in the tool of \S~\ref{sec:implementation-tool}.

The notion of knowledge needed to perform an attack adopted so far shifts the semantics load on the concept of secure channel. Besides its simplicity, this abstraction proves useful to model a great many different domains and lead to a sensible notion of attack tree. Nevertheless, it seems interesting to explore less abstract scenarios, where messages exchanged over channels do enjoy a structure and their content is exploitable in the continuation. There is a substantial corpus of literature on how to extend a process calculus to handle reasoning on terms (e.g., via equational theories or pattern matching, cf.~\cite{Vigo2013}), but at the semantic heart of such calculi lies the capability of testing if what is received matches what was expected.

In order to fully encompass the original Quality Calculus we should introduce both testing capabilities and structured messages. We limit here to show how to deal with the first extension, as it has a wider impact on the technical developments. As a matter of fact, distinguishing between a term $t$ and an expression $\SOME(t)$ we are already dealing with a (very simple) signature, and this gives the necessary insight onto our idea.

The syntax of the Value-Passing Quality Calculus, introduced in \S~\ref{sec:calculus-syntax}, is enhanced as follows. First of all, we allow now input and output channels to range over terms $t$, writing $\IN{t}{x}$ and $\OUT{t_1}{t_2}$. In particular, $t$ can be a variable $y$, realising name-passing. Second, we update the $\SYN{case}$ clause as $\CASEL{l}{x}{t}{P_1}{P_2}$, allowing to check the data payload (if any) of an input variable $x$. The semantics of \S~\ref{sec:calculus-semantics} is modified accordingly:
$$
\begin{array}{l}
\CASEL{l}{\SOME(c)}{c}{P_1}{P_2} \strans{\tau} P_1\\
\CASEL{l}{\SOME(c)}{y}{P_1}{P_2} \strans{\tau} P_1[c/y]\\
\CASEL{l}{\SOME(c)}{c'}{P_1}{P_2} \strans{\tau} P_2\ {\sf if}\ c \neq c'\\
\CASEL{l}{\NONE}{c}{P_1}{P_2} \strans{\tau} P_2\\
\CASEL{l}{\NONE}{y}{P_1}{P_2} \strans{\tau} P_2
\end{array}
$$
The translation from processes to formulae of \S~\ref{sec:discovering-translation} is lifted from propositional to first-order logic, so as to account for the richer expressiveness of the $\SYN{case}$ clause:
$$
\begin{array}{lll}
\trprot{\CASEL{l}{x}{t}{P_1}{P_2}}{\hp}	& = & \trprot{P_1}{(\hp \land \exists\FV{t}.(x = \SOME(t))}\ \cup\\
& &  \trprot{P_2}{(\hp \land \neg(\exists\FV{t}.(x = \SOME(t)))}\ \cup\\
& &  \{\hp \Rightarrow \overline{l}\}
\end{array}
$$
where $\SOME(\cdot)$ is a unary predicate, $\FV{t}$ denotes the variables free in $t$, and we write $x$ instead of $\overline{x}$ for now $x$ ranges over a set of optional data. Similarly, the translation of binders has now to record the term to which an input variable is bound when the corresponding binder is satisfied:
$$
\ths(\hp,\IN{t}{x}) = \{\exists y.(\hp \land t \Rightarrow (x = \SOME(y))\}
$$
where $t$ ranges over a set of data (the translation of output has to be updated similarly).

Finally, for building the tree some unification is needed in the backward-chaining search of \S~\ref{sec:displaying-inferring}:
$$
\begin{array}{l}
\trtree{\exists\FV{t}(x = \SOME(t))}{\mathcal{D}} = \trtree{\varphi\sigma}{\mathcal{D}}\\
\qquad\qquad {\sf\footnotesize where}\ \left(\exists\FV{t'}(\varphi \Rightarrow (x = \SOME(t')))\right)\, \in\, \rightsystem\ \land\ \exists\sigma.t = t'\sigma\\[3ex]
\trtree{\neg \exists\FV{t}(x = \SOME(t))}{\mathcal{D}} = \trtree{\neg \varphi\sigma}{\mathcal{D}}\\
\qquad\qquad {\sf\footnotesize where}\ \left(\exists\FV{t'}(\varphi \Rightarrow (x = \SOME(t')))\right)\, \in\, \rightsystem\ \land\ \exists\sigma.t = t'\sigma
\end{array}
$$
where $\sigma$ is a most general unifier.

We have thus shown how to lift all levels of the framework to name-passing calculi with full testing capabilities. From a high-level perspective, the extension allows inspecting how security checks are performed, while the basic developments consider checks as atomic entities, distinguishing between them through the cost map. There is the concrete risk, however, that the additional information available in the more detailed ``first-order'' trees would decrease readability drastically. In addition to this, whenever a finer-grained investigation is needed, we could take advantage of the modularity of the propositional framework, as discussed in \S~\ref{sec:discovering-modularity}.

Finally, in order to carry the extension to the Quality Protection Tool of \S~\ref{sec:implementation-tool}, the main extension would consist in introducing the unification of terms in the backward-chaining procedure. 

\end{document}